\documentclass[runningheads]{llncs}
\usepackage{amsmath,amssymb}
\usepackage[dvipsnames]{xcolor}

\usepackage{tikz}
\usetikzlibrary{automata}
\usetikzlibrary{shapes}
\usetikzlibrary{decorations.pathreplacing}
\usetikzlibrary{positioning}
\usetikzlibrary{shapes.multipart}

\usepackage{graphicx}
\usepackage[utf8]{inputenc}
\usepackage{multirow}
\usepackage{cleveref}
\usepackage{comment}
\usepackage{hhline}
\usepackage{mathtools}

\DeclareMathOperator{\Cost}{Cost}
\DeclareMathOperator{\Gain}{Gain}
\DeclareMathOperator{\CostOrGain}{f}
\DeclareMathOperator{\CostMin}{\Cost_{Min}}
\DeclareMathOperator{\GainMax}{\Gain_{Max}}
\DeclareMathOperator{\MaxCost}{MaxCost}
\DeclareMathOperator*{\argmax}{argmax}

\DeclareMathOperator{\Plays}{Plays}
\DeclareMathOperator{\Hist}{Hist}
\DeclareMathOperator{\Succ}{Succ}
\DeclareMathOperator{\First}{First}

\DeclareMathOperator{\Visit}{Visit}
\DeclareMathOperator{\SW}{SW}
\DeclareMathOperator{\Val}{Val}
\DeclareMathOperator{\ConstNE}{\Val}

\def\rest#1#2{#1_{\restriction#2}}
\def\outcome#1#2{\langle #1 \rangle_{#2}}

\def\arena{\mathcal{A}}

\def\extended{X}
\DeclareMathOperator{\extGame}{\mathcal{X}}
\def\extV{V^\extended}
\def\extE{E^\extended}
\def\extVi#1{V^\extended_{#1}}
\def\extFi#1{F^\extended_{#1}}
\def\extG{X}

\begin{document}
%
\title{On Relevant Equilibria in Reachability Games\thanks{Research partially supported by the PDR project “Subgame perfection in graph games” (F.R.S.-FNRS) and by  COST Action 16228 “GAMENET” (European Cooperation in Science and Technology).}}
%
%
\author{Thomas Brihaye\inst{1} \and Véronique Bruyère \inst{1} \and Aline Goeminne \inst{1,2} \and Nathan Thomasset \inst{1,3}}
\authorrunning{T. Brihaye et al.}
%

\institute{ Université de Mons (UMONS), Belgium\\
 \and Université libre de Bruxelles (ULB), Belgium\\
\and ENS Paris-Saclay, Université Paris-Saclay, France}


%
\maketitle              
\begin{abstract}

We study multiplayer reachability games played on a finite directed graph equipped with target sets, one for each player. In those reachability games, it is known that there always exists a Nash equilibrium (NE) and a subgame perfect equilibrium (SPE). But sometimes several equilibria may coexist such that in one equilibrium no player reaches his target set whereas in another one several players reach it. It is thus very natural to identify ``relevant'' equilibria. In this paper, we consider different notions of relevant equilibria including Pareto optimal equilibria and equilibria with  high social welfare. We provide complexity results for various related decision problems.

\keywords{ multiplayer non-zero-sum games played on graphs \and
reachability objectives \and
relevant equilibria \and
social welfare \and
Pareto optimality
}
\end{abstract}
%
%
%


\section{Introduction}

\emph{Two-player zero-sum games played on graphs} are commonly used to model \emph{reactive systems} where a system interacts with its environment \cite{PnueliR89}. In such setting the system wants to achieve a goal - to respect a certain property - and the environment acts in an antagonistic way. The underlying game is defined as follows: the two players are the system and the environment, the vertices of the graph are all the possible configurations in which the system can be and an infinite path in this graph depicts a possible sequence of interactions between the system and its environment. In such a game, each player chooses a \emph{strategy}: it is the way he plays given some information about the game and past actions of the other player. Following a strategy for each player results in a \emph{play} in the game. Finding how the system can ensure that a given property is satisfied amounts to find, if it exists, a \emph{winning strategy} for the system in this game. For some situations, this kind of model is too restrictive and a setting with more than two agents such that each of them has his own not necessarily antagonistic objective is more realistic. These games are called \emph{multiplayer non zero-sum games}. In this setting, the solution concept of winning strategy is not suitable anymore and different notions of \emph{equilibria} can be studied.

In this paper, we focus on \emph{Nash equilibrium} (NE)~\cite{nash50}: given a strategy for each player, no player has an incentive to deviate unilaterally from his strategy. We also consider the notion of \emph{subgame perfect equilibrium} (SPE) well suited for games played on graphs \cite{osbornebook}. We study these two notions of equilibria on \emph{reachability games}. In reachability games, we equip each player with a subset of vertices of the graph game that he wants to reach. We are interested in both the \emph{qualitative} and \emph{quantitative} settings. In the qualitative setting, each player only aims at reaching his target set, unlike the quantitative setting where each player wants to reach his target set as soon as possible.

It is well known that both NEs and SPEs exist in both qualitative and quantitative reachability games. But, equilibria such that no player reaches his target set and equilibria such that some players reach it may coexist. This observation has already been made in~\cite{Ummels08,Ummels06}.   In such a situation, one could prefer the second situation to the first one. In this paper, we study different versions of \emph{relevant equilibria}.

\subsubsection{Contributions}

For quantitative reachability games, we focus on the following three kinds of relevant equilibria: \emph{constrained equilibria}, \emph{equilibria optimizing social welfare} and \emph{Pareto optimal equilibria}. For constrained equilibria, we aim at minimizing the cost of each player \emph{i.e.,} the number of steps it takes to reach his target set (Problem 1). For equilibria optimizing social welfare, a player does not only want to minimize his own cost, he is also committed to maximizing the \emph{social welfare} (Problem 2). For Pareto optimal equilibria, we want to decide if there exists an equilibrium such that the tuple of the costs obtained by players following this equilibrium is \emph{Pareto optimal} in the set of all the possible costs that players can obtain in the game (Problem 3). We consider the decision variant of Problems 1 and 2; and the qualitative adaptations of the three problems. 

Our main contributions are the following.\emph{(i)} We study the complexity of the three decision problems. Our results gathered with previous works are summarized in Table~\ref{tab:sum}.\emph{(ii)} We characterize a sufficient finite-memory to solve the three decision problems. Our results and others from previous works are given in Table~\ref{tab:sum}.\emph{(iii)} We identify a subclass of reachability games in which there always exists an SPE where each player reaches his target set.\emph{(iv)} Given a play, we provide a \emph{characterization} which guarantees that this play is the outcome of an NE. This characterization is based on the values in the associated two-player zero-sum games called \emph{coalitional games}.

\begin{center}
\begin{table}
	\label{tab:sum}
	\caption{Complexity classes and memory results}
	\centering
	\begin{tabular}{ccc}
		\scalebox{0.8}{\begin{tabular}{|c||c|c||c|c|}
		    \hline
		    \multirow{2}{*}{Complexity} & \multicolumn{2}{c||}{Qual. Reach.}   & \multicolumn{2}{|c|}{Quant. Reach.} \\ \hhline{~||--||--}
		    & NE                       & SPE      & NE             & SPE               \\ \hline\hline
		    Prob. 1                   & NP-c \cite{condurache_et_al:LIPIcs:2016:6256}                    & PSPACE-c\cite{DBLP:journals/corr/abs-1809-03888} & NP-c           & PSPACE-c\cite{DBLP:journals/corr/abs-1905.00784}         \\ \hhline{-||--||--}
		    Prob. 2                   & NP-c                     & PSPACE-c & NP-c           & PSPACE-c          \\ \hhline{-||--||--}
		    Prob. 3                   & NP-h/$\Sigma^P_2$ & PSPACE-c &   NP-h/$\Sigma^P_2$              &        PSPACE-c           \\  \hline
		  \end{tabular}}
		
		& &
		
		\scalebox{0.8}{\begin{tabular}{|c||c|c||c|c|}
		\hline
		\multirow{2}{*}{Memory} & \multicolumn{2}{c||}{Qual. Reach.}   & \multicolumn{2}{|c|}{Quant. Reach.} \\ \hhline{~||--||--} 
		                            & NE                       & SPE      & NE             & SPE               \\ \hline\hline
		Prob. 1                   & Poly.\cite{condurache_et_al:LIPIcs:2016:6256}                   & Expo.\cite{DBLP:journals/corr/abs-1809-03888} & Poly.       & Expo.          \\  \hhline{-||--||--}
		Prob. 2                   & Poly.                    & Expo. & Poly.         & Expo.        \\  \hhline{-||--||--}
		Prob. 3                   &Poly. & Expo.  &   Poly.             &     Expo.              \\ \hline
		\end{tabular}}
		\end{tabular}
\end{table}
\end{center}

\subsubsection{Related work}

There are many results on NEs and SPEs played on graphs, we refer the reader to~\cite{Bruyere17} for a survey and an extended bibliography. We here focus on the results directly related to our contributions. 

Regarding  Problem 1, for NEs,  it is shown NP-complete only in the qualitative setting~\cite{condurache_et_al:LIPIcs:2016:6256}; for SPEs it is shown PSPACE-complete in both the qualitative and quantitative settings in \cite{DBLP:journals/corr/abs-1809-03888,DBLP:journals/corr/abs-1905.00784,concur}. Notice that in~\cite{Ummels08}, variants of Problem 1 for games with Streett, parity or co-Büchi winning conditions are shown NP-complete and decidable in polynomial time for Büchi. 

Regarding Problem 2, in the setting of games played on matrices, deciding the existence of an NE such that the expected social welfare is at most $k$ is NP-hard \cite{cs-GT-0205074}. Moreover, in~\cite{BouyerMS14} it is shown that deciding the existence of an NE which maximizes the social welfare is undecidable in concurrent games in which a cost profile is associated only with terminal nodes.

Regarding Problem 3, in the setting of  zero-sum two-player multidimensional mean-payoff games, the \emph{Pareto-curve} (the set of maximal thresholds that a player can force) is studied in~\cite{DBLP:conf/cav/BrenguierR15} by giving some properties on the geometry of this set. The autors provide a $\Sigma^P_2$ algorithm to decide if this set intersects a convex set defined by linear inequations.

Regarding the memory, in~\cite{DBLP:conf/lfcs/BrihayePS13} it is shown that there always exists an NE with memory at most $|V|+|\Pi|$ in quantitative reachability games, without any constraint on the cost of the NE. It is shown in~\cite{Ummels06} that, in multiplayer games with $\omega$-regular objectives, there exists an SPE with a given payoff if and only if there exists an SPE with the same payoff but with finite memory. Moreover, in~\cite{DBLP:journals/corr/abs-1809-03888} it is claimed that it is sufficient to consider  strategies with an exponential memory to solve Problem 1 for SPE in qualitative reachability games.

Finally, we can find several kinds of outcome characterizations for Nash equilibria and variants, \emph{e.g.,} in multiplayer games equipped with prefix-linear cost functions and such that the vertices in coalitional games have a value (summarized in \cite{Bruyere17}), in multiplayer games with prefix-independent Borel objectives \cite{Ummels08}, in multiplayer games with classical $\omega$-regular objectives (as reachability) by checking if there exists a play which satisfies an LTL formula \cite{condurache_et_al:LIPIcs:2016:6256}, in concurrent games \cite{haddad}, etc. Such characterizations are less widespread for subgame perfect equilibria, but one can recover one for quantitative reachability games thanks to a value-iteration procedure \cite{DBLP:journals/corr/abs-1905.00784}.

\subsubsection{Structure of the paper} Due to the lack of space, we decide to only detail results for quantitative reachability games while results for qualitative reachability games are only summarized in Table~\ref{tab:sum}. In Section~\ref{section:preliminaries}, we introduce the needed background and define the different studied problems. In Section~\ref{sec:existence}, we identify families of reachability games for which
there always exists a relevant equilibrium, for different notions of relevant equilibrium. In Section~\ref{sec:decision}, we provide the main ideas necessary to obtain our complexity results (see Table~\ref{tab:sum}). The detailed proofs for the quantitative reachability setting, together with additional results on qualitative reachability games are provided in the appendices.

\section{Preliminaries and studied problems}
\label{section:preliminaries}

\subsubsection{Arena, game and strategies} 
An \emph{arena} is a tuple $\arena = (\Pi,V,E,(V_i)_{i\in \Pi})$ such that: \emph{(i)} $\Pi$ is a finite set of players; \emph{(ii)} $V$ is a finite set of vertices; \emph{(iii)}  $E \subseteq V \times V$ is a set of edges such that for all $v \in V$ there exists $v'\in V$ such that $(v,v')\in E$ and \emph{(iv)} $(V_i)_{i\in \Pi}$ is a partition of $V$ between the players.


A \emph{play} in $\arena$ is an infinite sequence of vertices $\rho = \rho_0 \rho_1 \ldots$ such that for all $k \in \mathbb{N}$, $(\rho_k, \rho_{k+1}) \in E$.  A \emph{history} is a finite sequence $h = h_0h_1 \ldots h_k$ with $k \in\mathbb{N}$ defined similarly. The \emph{length} $|h|$ of $h$ is the number $k$ of its edges. We denote the set of plays by $\Plays$ and the set of histories by $\Hist$. Moreover, the set $\Hist_i$ is the set of histories such that their last vertex $v$ is a vertex of player $i$, i.e. $v \in V_i$. 

Given a play $\rho \in \Plays$ and $k \in \mathbb{N}$, the \emph{prefix} $\rho_0 \rho_1 \ldots \rho_k$ of $\rho$ is denoted by $\rho_{\leq k}$ and its \emph{suffix} $\rho_k \rho_{k+1} \ldots$ by $\rho_{\geq k}$. A play $\rho$ is called a \emph{lasso} if it is of the form $\rho = h\ell^\omega$ with $h\ell \in \Hist$. Notice that $\ell$ is not necessarily a simple cycle. The \emph{length} of a lasso $h\ell^\omega$ is the length of $h\ell$. 

A \emph{game} $\mathcal{G} = (\arena,(\Cost_i)_{i\in\Pi})$ is an arena equipped with a cost function profile $ (\Cost_i)_{i\in\Pi}$ such that for all $i \in \Pi$, $\Cost_i : \Plays \rightarrow \mathbb{N} \cup \{+\infty\}$ is a \emph{cost function}  which assigns a cost to each play $\rho$ for player $i$. We also say that the play $\rho$ has \emph{cost profile}  $(\Cost_i(\rho))_{i\in\Pi}$. Given two cost profiles $c,c' \in (\mathbb{N}\cup \{+\infty\})^{|\Pi|}$, we say that $c \leq c'$ if and only if for all $i \in \Pi$, $c_i \leq c'_i$.

An initial vertex $v_0\in V$ is often fixed, and we call $(\mathcal{G}, v_0)$ an \emph{initialized game}. A play (resp. a history) of $(\mathcal{G},v_0)$ is then a play (resp. a history) of $\mathcal{G}$ starting in $v_0$. The set of such plays (resp. histories) is denoted by $\Plays(v_0)$ (resp. $\Hist(v_0)$). The notation $\Hist_i(v_0)$ is used when these histories end in a vertex $v \in V_i$.

Given a game $\mathcal G$, a \emph{strategy} for player $i$ is a function $\sigma_i: \Hist_i \rightarrow V$. It assigns to each history $hv$, with $v \in V_i$, a vertex $v'$ such that $(v,v') \in E$. In an initialized game $(\mathcal{G},v_0)$, $\sigma_i$ needs only to be defined for histories starting in $v_0$. We denote by $\Sigma_i$ the set of strategies for Player $i$. A play $\rho=\rho_0\rho_1\ldots$ is \emph{consistent} with  $\sigma_i$ if for all $\rho_k \in V_i$, $\sigma_i(\rho_0 \ldots \rho_k) = \rho_{k+1}$. A strategy $\sigma_i$ is \emph{positional} if it only depends on the last vertex of the history, \emph{i.e.}, $\sigma_i(hv) = \sigma_i(v)$ for all $hv \in \Hist_i$.  It is \emph{finite-memory} if it can be encoded by a finite-state machine.

A \emph{strategy profile} is a tuple $\sigma = (\sigma_i)_{i\in \Pi}$ of strategies, one for each player. 
Given an initialized game $(\mathcal{G}, v_0)$ and a strategy profile $\sigma$, there exists an unique play from $v_0$ consistent with each strategy $\sigma_i$. We call this play the \emph{outcome} of $\sigma$ and denote it by $\outcome{\sigma}{v_0}$. We say that $\sigma$ has cost profile $(\Cost_i(\outcome{\sigma}{v_0}))_{i\in\Pi}$.


\subsubsection{Quantitative reachability games}

In this article, we are interested in \emph{reachability games}: each player has a target set of vertices that he wants to reach. 
   
\begin{definition} \label{def:quantitativeGame}
		A \emph{quantitative reachability game} $\mathcal{G} = (\mathcal{A}, (\Cost_i)_{i\in \Pi}, (F_i)_{i\in \Pi})$ is a game enhanced with a target set $F_i \subseteq V$ for each player $i \in \Pi$ and for all $i \in \Pi$ the cost function $ \Cost_i$ is defined as follows: for all $\rho= \rho_0\rho_1\ldots \in \Plays$: $\Cost_i(\rho) = k$ if  $k\in \mathbb{N}$ is the least index such that $\rho_k \in F_i$ and $\Cost_i(\rho) = +\infty$ if such index does not exist.
\end{definition}

 In \emph{quantitative reachability games}, players have to pay a cost equal to the number of edges until visiting their own target set or $+\infty$ if it is not visited. Thus each player aims at \emph{minimizing} his cost.

\subsubsection{Solution concepts}
\label{section:solutionConcepts}

In the multiplayer game setting, the solution concepts usually studied are \emph{equilibria}. We recall the concepts of Nash equilibrium and subgame perfect equilibrium.

Let $\sigma = (\sigma_i)_{i\in \Pi}$ be a strategy profile in an initialized game $(\mathcal{G},v_0)$. When we highlight the role of player~$i$, we denote $\sigma$ by $(\sigma_i, \sigma_{-i})$ where $\sigma_{-i}$ is the profile $(\sigma_j)_{j\in \Pi \setminus \{i\}}$. A strategy $\sigma'_i \neq \sigma_i$ is a \emph{deviating} strategy of Player~$i$, and it is a \emph{profitable deviation} for him if $\Cost_i(\outcome{\sigma}{v_0}) > \Cost_i(\outcome{\sigma'_i, \sigma_{-i}}{v_0})$.

The notion of Nash equilibrium is classical: a strategy profile $\sigma$ in an initialized game $(\mathcal{G},v_0)$ is a \emph{Nash equilibrium} (NE) if no player has an incentive to deviate unilaterally from his strategy, i.e. no player has a profitable deviation.

\begin{definition}[Nash equilibrium]
	Let $(\mathcal{G},v_0)$ be an initialized quantitative reachability game. The strategy profile $\sigma$ is an NE if for each $i \in \Pi$ and each deviating strategy $\sigma'_i$ of Player $i$, we have $\Cost_i(\outcome{\sigma}{v_0}) \leq \Cost_i(\outcome{\sigma'_i, \sigma_{-i}}{v_0})$. 
\end{definition}

When considering games played on graphs, a useful refinement of NE is the concept of \emph{subgame perfect equilibrium} (SPE) which is a strategy profile that is an NE in each subgame. Formally, given a game ${\mathcal G} = (\arena,(\Cost_i)_{i\in\Pi})$, an initial vertex $v_0$, and a history $hv \in \Hist(v_0)$, the initialized game $(\rest{\mathcal{G}}{h},v)$ such that $\rest{\mathcal{G}}{h} = (\arena, (\Cost_{i\restriction h})_{i\in\Pi})$ where $\Cost_{i\restriction h}(\rho) = \Cost_i(h\rho)$ for all $i \in \Pi$ and $\rho \in V^{\omega}$ is called a \emph{subgame} of $(\mathcal{G},v_0)$. Notice that $(\mathcal{G},v_0)$ is a subgame of itself. Moreover if $\sigma_i$ is a strategy for player~$i$ in $(\mathcal{G},v_0)$, then $\sigma_{i\restriction h}$ denotes the strategy in $(\rest{\mathcal{G}}{h},v)$ such that for all histories $h'\in \Hist_i(v)$, $\sigma_{i\restriction h}(h') = \sigma_i(hh')$. Similarly, from a strategy profile $\sigma$ in $(\mathcal{G},v_0)$, we derive the strategy profile $\rest{\sigma}{h}$ in $(\rest{\mathcal{G}}{h},v)$. 

\begin{definition}[Subgame perfect equilibrium]
	Let $(\mathcal{G},v_0)$ be an initialized game. A strategy profile $\sigma$ is an SPE in $(\mathcal{G},v_0)$ if for all $hv \in \Hist(v_0)$, $\rest{\sigma}{h}$ is an NE in $(\rest{\mathcal{G}}{h},v)$.
\end{definition}

Clearly, any SPE is an NE and it is stated in Theorem 2.1 in \cite{DBLP:journals/corr/abs-1205-6346} that there always exists an SPE (and thus an NE) in quantitative reachability games.


\subsubsection{Studied problems}

We conclude this section with the problems studied in this article. Let us first recall the concepts of social welfare and Pareto optimality. Let $(\mathcal{G},v_0)$ be an initialized  quantitative reachability game with $\mathcal{G} = (\mathcal{A}, (\Cost_i)_{i\in \Pi}, (F_i)_{i\in \Pi})$.  Given $\rho = \rho_0\rho_1\ldots \in \Plays(v_0)$,  we denote by $\Visit(\rho)$ the set of players  who visit their target set along $\rho$, \emph{i.e.,} $\Visit(\rho) = \{ i \in \Pi \mid \text{there exists } n \in \mathbb{N} \text{ st. } \rho_n\in F_i \}$.\footnote{We can easily adapt this definition to histories.} The \emph{social welfare} of $\rho$, denoted by $\SW(\rho)$, is the pair $(|\Visit(\rho)|,\sum_{i \in \Visit(\rho)} \Cost_i(\rho))$. Note that it takes into account both the number of players who visit their target set and their accumulated cost to reach those sets. Finally, let $P = \{ (\Cost_i(\rho))_{i \in \Pi} \mid \rho \in \Plays(v_0) \}   \subseteq (\mathbb{N} \cup \{+\infty\})^{|\Pi|}$. A cost profile $p \in P$ is \emph{Pareto optimal in $\Plays(v_0)$} if it is minimal in $P$ with respect to the componentwise ordering $\leq$ on $P$ \footnote{For convenience, we prefer to say that $p$ is Pareto optimal in $\Plays(v_0)$ rather than in $P$.}.

Let us now state the studied decision problems. The first two problems are classical: they ask whether there exists a solution (NE or SPE) $\sigma$ satisfying certain requirements that impose  bounds on either  $(\Cost_i(\outcome{\sigma}{v_0}))_{i \in \Pi}$ or on $\SW(\outcome{\sigma}{v_0})$. 
%
%

\begin{problem}[Threshold decision problem]
\label{prob:decisionThreshold}
Given an initialized quantitative reachability game $(\mathcal{G},v_0)$, given a threshold $y \in (\mathbb{N} \cup \{+\infty\})^{|\Pi|}$, decide whether there exists a solution $\sigma$ such that $(\Cost_i(\outcome{\sigma}{v_0}))_{i \in \Pi} \leq y$.
\end{problem}
The most natural requirements are to impose upper bounds on the costs that the players have to pay and no lower bounds. One might also be interested in imposing an interval $[x_i,y_i]$ in which must lie the cost paid by Player~$i$. 

In the second problem, constraints are imposed on the social welfare, with the aim to maximize it. We use the lexicographic ordering on $\mathbb{N}^{2}$ such that $(k,c) \succeq (k',c')$ if and only if \emph{(i)} $k \geq k'$ or \emph{(ii)} $k = k'$ and $c \leq c'$. 

\begin{problem}[Social welfare decision problem]
\label{prob:decisionWelfare}
Given an initialized quantitative reachability game $(\mathcal{G},v_0)$, given two thresholds $k \in \{0, \ldots, |\Pi| \}$ and $c \in \mathbb{N}$, decide whether there exists a solution $\sigma$ such that $\SW(\outcome{\sigma}{v_0}) \succeq (k,c)$.
\end{problem}

Notice that with the lexicographic ordering, we want to first maximize the number of players who visit their target set, and then to minimize the accumulated cost to reach those sets. Let us now state the last studied problem.

\begin{problem}[Pareto optimal decision problem]
\label{prob:decisionPareto}
Given an initialized  quantitative reachability game $(\mathcal{G},v_0)$ decide whether there exists a solution $\sigma$ in $(\mathcal{G},v_0)$ such that $(\Cost_i(\outcome{\sigma}{v_0}))_{i\in\Pi}$ is Pareto optimal in $\Plays(v_0)$.
\end{problem}

\begin{remark}
Problems~\ref{prob:decisionThreshold} and~\ref{prob:decisionWelfare} impose constraints with \emph{large} inequalities. We could also impose strict inequalities or even a mix of strict and large inequalities. The results of this article can be easily adapted to those variants.
\end{remark}

We conclude this section with an illustrative example.

\begin{example} \label{ex:1}
Consider the quantitative reachability game $(\mathcal{G},v_0)$ of Figure~\ref{fig:2players}. We have two players such that the vertices of Player~$1$ (resp. Player~$2$) are rounded (resp. rectangular) vertices. For the moment, the reader should not consider the value indicated on the right of the vertices' labeling. Moreover $F_1 = \{v_3,v_4\}$ and $F_2 = \{v_1,v_4\}$.  In this figure, an edge $(v,v')$ labeled by $x$ should be understood as a path from $v$ to $v'$ with length $x$. Observe that $F_1$ and $F_2$ are both reachable from the initial vertex $v_0$. Moreover the two Pareto optimal cost profiles are $(3,3)$ and $(2,6)$: take a play with prefix $v_0v_2v_4$ in the first case, and a play with prefix $v_0v_2v_3v_0v_1$ in the second case. 

\begin{figure}
 \centering
 \scalebox{0.8}{\begin{tikzpicture}[-latex, auto, node distance = 1.5 cm and 1.5 cm, on grid, semithick, round/.style = {draw,rounded corners=6pt}, sq/.style = {draw,rectangle}, di/.style = {draw,diamond}]
  \node[sq] (0) at (0,0){$v_0$: $\mathbf{3}$};
  \node[round] (1) at (-2,0) {$v_1$: $\mathbf{+\infty}$};
  \node[round] (2) at (2,0) {$v_2$: $\mathbf{1}$};
  \node[round] (3) at (2,1) {$v_3$: $\mathbf{0}$};
  \node[round] (4) at (4,0) {$v_4$: $\mathbf{0}$};

  \draw[->] (0) -- node[below] {$3$} (1);
  \draw[->] (0) -- (2);
  \draw[->] (2) -- (3);
  \draw[->] (2) -- node[below] {$2$} (4);
  \draw[->] (3) to [bend right] (0);
  \draw[->] (4) to [bend left] (0);
  \draw[->] (1) to [bend left] (0);
 \end{tikzpicture}}
 \caption{A two-player quantitative reachability game such that $F_1=\{v_3,v_4\}$ and $F_2 =\{v_1,v_4\}$}
 \label{fig:2players}
\end{figure}
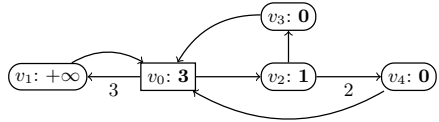

For this example, we claim that there is no NE (and thus no SPE) such that its cost profile is Pareto optimal (see Problem~\ref{prob:decisionPareto}). Assume the contrary and suppose that there exists an NE $\sigma$ such that its outcome $\rho$ has cost profile $(3,3)$, meaning that $\rho$ begins with $v_0v_2v_4$. Then Player~$1$ has a profitable deviation such that after history $v_0v_2$ he goes to $v_3$ instead of $v_4$ in a way to pay a cost of $2$ instead of $3$, which is a contradiction. Similarly assume that there exists an NE $\sigma$ such that its outcome $\rho$ has cost profile $(2,6)$, meaning that $\rho$ begins with $v_0v_2v_3v_0v_1$. Then Player~$2$ has a profitable deviation such that after history $v_0$ he goes to $v_1$ instead of $v_2$, again a contradiction. So there is no NE $\sigma$ in $(\mathcal{G},v_0)$ such that $(\Cost_i(\outcome{\sigma}{v_0}))_{i\in\Pi}$ is Pareto optimal in $\Plays(v_0)$.

The previous discussion shows that there is no NE $\sigma$ such that $(0,0) = x \leq (\Cost_i(\outcome{\sigma}{v_0}))_{i \in \Pi} \leq y = (3,3)$ (see Problem~\ref{prob:decisionThreshold}). This is no longer true with $y = (6,3)$. Indeed, one can construct an NE $\tau$ whose outcome has prefix $v_0v_1v_0v_2v_3$ and cost profile $(6,3)$. This also shows that there exists an NE $\sigma$ (the same $\tau$ as before) that satisfies $\SW(\outcome{\sigma}{v_0}) \succeq (k,c) = (2,9)$ (with $\tau$ both players visit their target set and their accumulated cost to reach it equals $9$).
\qed\end{example}


\section{Existence problems} \label{sec:existence}

In this section, we show that for particular families of reachability games and requirements, there is no need to solve the related decision problems because they always have a positive answer in this case. 

We begin with the family constituted by all reachability games with a \emph{strongly connected} arena. The next theorem then states that there always exists a solution that visits all non empty target sets.

\begin{theorem} \label{thm:visitAll}
Let $(\mathcal{G},v_0)$ be an initialized quantitative reachability game such that its arena $\arena$ is strongly connected. There exists an SPE $\sigma$ (and thus an NE) such that its outcome $\outcome{\sigma}{v_0}$ visits all target sets $F_i$, $i \in \Pi$, that are non empty.
\end{theorem}

Let us comment this result. For this family of games, the answer to Problem~\ref{prob:decisionThreshold} is always positive for particular thresholds. In case of quantitative reachability, take strict constraints $< y_i = +\infty$ if $F_i \neq \emptyset$ and large constraints $\leq +\infty$ otherwise. The answer to Problem~\ref{prob:decisionWelfare} is also always positive for threshold $k = |\{ i \mid F_i \neq \emptyset \}|$ and $c = +\infty$. 

In the statement of Theorem~\ref{thm:visitAll}, as the arena is strongly connected, $F_i$ is non empty if and only if $F_i$ is reachable from $v_0$. Also notice that the hypothesis that the arena is strongly connected is necessary. Indeed, it is easy to build an example with two players (Player $1$ and Player $2$) such that from $v_0$ it is not possible to reach both $F_1$ and $F_2$.

We now turn to the second result of this section. The next theorem states that even with  only two players there exists an initialized quantitative reachability game that has no NE with a cost profile which is Pareto optimal. To prove this result, we  only have to come back to the quantitative reachability game  of Figure~\ref{fig:2players}. We explained in Example~\ref{ex:1} that there is no NE in this game such that its cost profile is Pareto optimal.

\begin{theorem} \label{thm:existencePareto}
There exists an initialized quantitative reachability game with $|\Pi| = 2$ that has no NE with a cost profile which is Pareto optimal in $\Plays(v_0)$.
\end{theorem}

Notice that in the qualitative setting, in two-player games, there always exists an NE (resp. SPE) such that the gain profile\footnote{In the qualitative setting, each player obtain a gain that he wants to maximize: either 1 (if he visits his target set) or 0 (otherwise), all definitions are adapted accordingly.} is Pareto optimal in $\Plays(v_0)$ however this existence result cannot be extended to three players.


\section{Solving decision problems} \label{sec:decision}


In this section, we provide the complexity results for the different problems without any assumption on the arena of the game. Even if we provide complexity lower bounds, the main part of our contribution is to give the upper bounds. Roughly speaking the decision algorithms work as follows: they guess a path and check that it is the outcome of an equilibrium satisfying the relevant property (such as Pareto optimality). In order to verify that a path is an equilibrium outcome, we rely on the outcome characterization of equilibria, presented in Section~\ref{subsection:charac}. These characterizations rely themselves on the notion of $\lambda$-consistent play, introduced in Section~\ref{section:consistentPlay}. As the guessed path should be finitely representable, we show that we can only consider $\lambda$-consistent lassoes, in Section~\ref{section:lassoes}. Finally, we expose the philosophy of the algorithms providing the upper bounds on the complexity of the three problems in Section~\ref{section:algo}.

\subsection{$\lambda$-consistent play}
\label{section:consistentPlay}

We here define the \emph{labeling function}, $\lambda : V \rightarrow \mathbb{N}\cup \{+\infty \}$ used to obtain the outcome characterization of equilibria. Given a vertex $v \in V$ along a play $\rho$, intuitively, the value $\lambda(v)$ represents the maximal number of steps within which the player who owns this vertex should reach his target set along $\rho$ starting from $v$. A play which satisfies the constraints given by $\lambda$ is called a $\lambda$-consistent play. 

\begin{definition}[$\lambda$-consistent play] Let $(\mathcal{G},v_0)$ be a quantitative reachability game and $\lambda : V \rightarrow \mathbb{N}\cup \{+\infty \}$ be a labeling function. Let $\rho \in \Plays$ be a play,  we say that $\rho = \rho_0 \rho_1\ldots$ is $\lambda $-consistent if for all $i \in \Pi$ and all $k \in \mathbb{N}$ such that $i \not \in \Visit(\rho_0\ldots \rho_k)$ and $\rho_k \in V_i$:  $ \Cost_i(\rho_{\geq k}) \leq \lambda(\rho_k). \label{eq:constraintLambdaConst}$
\end{definition}
The link between $\lambda$-consistency and equilibrium is made in Section~\ref{subsection:charac}.
\begin{example}\label{ex:lambda}
	Let us come back to Example~\ref{ex:1} and assume that the values indicated on the right of the vertices' labeling represent the valuation of a labeling function $\lambda$. Let us first consider the play $\rho=(v_0v_2v_4)^\omega$ with cost profile $(3,3)$. We have that $\Cost_2(\rho) = 3 \leq \lambda(v_0) = 3$ but $ \Cost_1(\rho_{\geq 1}) = \Cost_1(v_2v_4(v_0v_2v_4)^\omega) = 2 > \lambda(v_2) =1$. This means that $(v_0v_2v_4)^\omega$ is not $\lambda$-consistent. Secondly, one can easily see that the play $v_0v_1(v_0v_2v_3)^\omega$ is $\lambda$-consistent. 
\end{example}




\subsection{Characterizations}
\label{subsection:charac}

\subsubsection{Outcome characterization of Nash equilibria}

To define the labeling function $\lambda$ which allows us to obtain this characterization, we need to study the rational behavior of one player playing against the \emph{coalition} of the other players. In order to do so, with a quantitative reachability game $\mathcal{G} = (\arena, (\Cost_i)_{i\in \Pi}, (F_i)_{i\in \Pi})$, we can associate $|\Pi|$ \emph{two-player zero-sum quantitative games} \cite{DBLP:conf/lfcs/BrihayePS13}. For each $i \in \Pi$, we depict by $\mathcal{G}_i$ the \emph{(quantitative) coalitional game} associated with Player $i$. In such a game Player $i$ (which becomes Player $Min$) wants to reach the target set $F=F_i$ within a minimum number of steps, and the coalition of all players except Player $i$ (which forms one player called Player $Max$, aka $-i$) aims to avoid it or, if it is not possible, maximize the number  of steps until reaching $F$.

Given a coalitional game $\mathcal{G}_i$ and a vertex $v \in V$, the \emph{value} of $\mathcal{G}_i$ from $v$, depicted by $\Val_i(v)$, allows us to know what is the lowest (resp. greatest) cost (resp. gain) that Player $Min$ (resp. Player $Max$) can ensure to obtain from $v$.  Moreover, as quantitative coalitional games are determined these values always exist and can be computed in polynomial time~\cite{DBLP:conf/lfcs/BrihayePS13,DBLP:journals/acta/BrihayeGHM17,Khachiyan2008}.

An \emph{optimal strategy} for Player $Min$ (resp. Player $Max$) in a coalitional game $\mathcal{G}_i$ is a strategy which ensures that, from all vertex $v \in V$, Player $Min$ (resp. Player $Max$) will pay (resp. obtain) at most $\Val_i(v)$ by following this strategy whatever the strategy of the other player. For each $i \in \Pi$, we know that  there always exist optimal strategies for both players in $\mathcal{G}_i$. Moreover, we can always found optimal strategies which are positional \cite{DBLP:conf/lfcs/BrihayePS13}. \\

In our characterization, we show that the outcomes of NEs are exactly the plays which are $\Val$-consistent, with the labeling function $\Val$ defined in this way: for all $v \in V$, $ \ConstNE(v) = \Val_i(v) \quad \text{ if } v \in V_i.$

\begin{theorem}[Characterization of NEs]
	\label{thm:critOutcomeNE}
	Let $(\mathcal{G},v_0)$ be a quantitative reachability game and let $\rho \in \Plays(v_0)$ be a play, the next assertions are equivalent:
	
	\begin{enumerate}
		\item there exists an NE $\sigma$ such that $\outcome{\sigma}{v_0}= \rho$; \label{state:1critOutcomeNE}
		\item the play $\rho$ is $\ConstNE$-consistent. \label{item:critOut2}
	\end{enumerate}	
Additionally, if $\rho = h\ell^\omega$ is a lasso, we can replace the first item by: there exists an NE $\sigma$ with memory in $\mathcal{O}(|h\ell|+|\Pi|)$ and such that $\outcome{\sigma}{v_0}= \rho$. 
\end{theorem}

The main idea is that if the second assertion is false, then there exists a player $i$ who has an incentive to deviate along $\rho$. Indeed, if there exists $k \in \mathbb{N}$ such that $\Cost_i(\rho_{\geq k}) > \Val_i(\rho_k)$ ($\rho_k \in V_i$) it means that Player $i$ can ensure a better cost for him even if the other players play in coalition and in an antagonistic way. Thus, Player $i$ has a profitable deviation. For the second implication, the Nash equilibrium $\sigma$ is defined as follows: all players follow the outcome $\rho$ but if one player, assume it is Player $i$, deviates from $\rho$ the other players form a coalition $-i$ and punish the deviator by playing the optimal strategy of player $-i$ in the coalitional game $\mathcal{G}_i$. Thus, if $\rho=h\ell^\omega$, a player has to remember: \emph{(i)} $h\ell$ to know both what he has to play 
and if someone has deviated  and \emph{(ii)} who is the deviator. 

\begin{example}
Let us go back to Example~\ref{ex:lambda}, in this example the used labeling function $\lambda$ is in fact the labeling function $\Val$. We proved in Example~\ref{ex:lambda} that the play $(v_0v_2v_4)^\omega$ is not $\Val$-consistent and so not the outcome of an NE by Theorem~\ref{thm:critOutcomeNE}. On the contrary, we have seen that the play $v_0v_1(v_0v_2v_3)^\omega$ is $\Val$-consistent and it means that it is the outcome of an NE (again by Theorem~\ref{thm:critOutcomeNE}). Notice that we have already proved these two facts in Example~\ref{ex:1}. 
\end{example}


\subsubsection{Outcome characterization of subgame perfect equilibria}
In the previous section, we proved that the set of plays which are $\ConstNE$-consistent is equal to the set of outcomes of NEs. We now want to have the same kind of characterization for SPEs. We may not use the notion of $\ConstNE$-consistent plays because there exist plays which are $\ConstNE$-consistent but which are not the outcome of an SPE. But,  we can recover the characterization of SPEs thanks to a different labeling function defined in \cite{DBLP:journals/corr/abs-1905.00784} that we depict by $\lambda^*$. Notice that, $\lambda^*$ is not defined on the vertices of the game $\mathcal{G}$ but on the vertices of the \emph{extended game} $\mathcal{X}$ associated with $\mathcal{G}$. Vertices in such a game are the vertices in $\mathcal{G}$ equipped with a subset of players who have already visited their target set. This game is also a reachability game thus all concepts and definitions introduced in Section~\ref{section:preliminaries} hold. Moreover, there is a one-to-one correspondence between SPEs in $\mathcal{G}$ and its extended game. This is the reason why we solve the different decision problems on the extended games $(\mathcal{X},x_0)$, where $x_0 = (v_0, \Visit(v_0))$, instead of $(\mathcal{G},v_0)$. More details are given in  \cite{DBLP:journals/corr/abs-1905.00784}. However, it is very important to notice that some of our results depend on $|V|$  (resp. $|\Pi|$) that are the number of vertices (resp. players) in $\mathcal{G}$ and not in $\mathcal{X}$.


\begin{theorem}[\cite{DBLP:journals/corr/abs-1905.00784} Characterization of SPEs]
	\label{thm:critOutcomeSPE}
	Let $(\mathcal{G},v_0)$ be a quantitative reachability game and $(\mathcal{X},x_0)$ be its extended game and let $\rho = \rho_0\rho_1 \ldots \in \Plays(x_0)$ be a play in the extended game, the next assertions are equivalent:
	
	\begin{enumerate}
		\item there exists a subgame perfect equilibrium $\sigma$ such that $\outcome{\sigma}{x_0}= \rho$; \label{state:1critOutcomeSPE}
		\item the play $\rho$ is $\lambda^*$-consistent. \label{item:critOut2SPE}
	\end{enumerate}
\end{theorem}


\subsection{Sufficiency of lassoes}
\label{section:lassoes}

In this section, we provide technical results which given a $\lambda$-consistent play produce an associated $\lambda$-consistent lasso. In the sequel, we show that working with these lassoes is sufficient for the algorithms. 

The associated lassoes are built by eliminating some \emph{unnecessary cycles} and then identifying a prefix $h\ell$ such that $\ell$ can be repeated infinitely often. An unnecessary cycle is a cycle inside of which no new player visits his target set. More formally, let $\rho = \rho_0\rho_1 \ldots \rho_k \ldots \rho_{k+\ell} \ldots$ be a play in $\mathcal{G}$, if $\rho_k=\rho_{k+\ell}$  and $\Visit(\rho_0\ldots \rho_k) = \Visit(\rho_0\ldots \rho_{k+\ell})$ then the cycle $\rho_k \ldots \rho_{k+\ell}$ is called an unnecessary cycle.

We call: (P1) the procedure which eliminates an unnecessary cycle, \emph{i.e.,} let $\rho = \rho_0\rho_1 \ldots \rho_k \ldots \rho_{k+\ell} \ldots$ such that $\rho_k \ldots \rho_{k+\ell}$ is an unnecessary cycle, $\rho$ becomes $\rho'= \rho_0\ldots \rho_k \rho_{k+\ell+1} \ldots$ and (P2) the procedure which turns $\rho$ into a lasso $\rho' = h\ell^\omega$ by copying $\rho$ long enough for all players to visit their target set and then to form a cycle after the last player has visited his target set. If no player visits his target set along $\rho$, then (P2) only copies $\rho$ long enough to form a cycle. Notice that, given $\rho\in \Plays$, applying (P1) or (P2) may involve a decreasing of the costs but for both $\Visit(\rho)= \Visit(\rho')$ and for (P2) $\Visit(h)= \Visit (\rho')$. Additionally, applying (P1) until it is no longer possible and then (P2), leads to a lasso with length at most $(|\Pi|+1)\cdot |V|$ and cost less than or equal to $|\Pi|\cdot|V|$ for players who have visited their target set.

Additionally, applying (P1) or (P2) on  $\lambda$-consistent  play preserves this property. It is stated in Lemma~\ref{lem:LassoesConsis} which is in particular true for extended games.

\begin{lemma}
	\label{lem:LassoesConsis}
	Let $(\mathcal{G},v_0)$ be a quantitative reachability game and $\rho \in \Plays$ be a $\lambda$-consistent play for a given labeling function $\lambda$. If $\rho'$ is the play obtained by applying (P1) or (P2) on $\rho$, then $\rho'$ is $\lambda$-consistent.
\end{lemma}


These properties on (P1) and (P2) allow us to claim that it is sufficient to deal with lassoes with polynomial length to solve  Problems~\ref{prob:decisionThreshold} and~\ref{prob:decisionPareto} for NEs and it give us some bounds on the needed memory and the costs for each problem.

\begin{corollary}
	\label{cor:corCritOut2ConstraintProb}
	Let $\sigma$ be an NE (resp. SPE) in a quantitative reachability game $(\mathcal{G},v_0)$ (resp. $(\mathcal{X},x_0)$ its extended game) and $y \in (\mathbb{N}\cup\{+\infty\})^{|\Pi|}$. Let $w_0 = v_0$ (resp. $w_0= x_0$).
	If $ (\Cost_i(\outcome{\sigma}{w_0}))_{i\in\Pi} \leq y$, then there exists $\tau$ an NE (resp. SPE) in $(\mathcal{G},v_0)$ (resp. $(\mathcal{X},x_0)$) such that:
	\begin{itemize}
		\item  $(\Cost_i(\outcome{\tau}{w_0}))_{i\in\Pi} \leq y$;
		\item  $\outcome{\tau}{w_0}$ is a lasso $h\ell^\omega$ such that $|h\ell|\leq (|\Pi|+1)\cdot |V| $;
		\item for each $i \in \Visit(\outcome{\tau}{w_0})$, $\Cost_i(\outcome{\tau}{w_0}) \leq |\Pi|\cdot |V|$;
		 \item $\tau$ has  memory in $\mathcal{O}((|\Pi|+1)\cdot |V|)$  (resp. $\mathcal{O}(2^{|\Pi|} \cdot |\Pi|\cdot |V|^{(|\Pi|+1)\cdot(|\Pi|+|V|)+1})$).
		\end{itemize}
\end{corollary}

\begin{proposition}
	\label{prop:particularPareto}
	Let $(\mathcal{G},v_0)$ (resp. $(\mathcal{X},x_0)$ its extended game) be a quantitative reachability game and let $\sigma$ be an NE (resp. SPE). Let $w_0 = v_0$ (resp. $w_0 = x_0$). If we have that $(\Cost_i(\outcome{\sigma}{w_0}))_{i\in\Pi}$ is Pareto optimal in $\Plays(w_0)$, then:
	
	\begin{itemize}
		\item for all $i \in \Visit(\outcome{\sigma}{w_0})$, $\Cost_i(\outcome{\sigma}{w_0}) \leq |V|\cdot|\Pi|;$
		\item there exists $\tau$ an NE (resp. SPE) such that $\outcome{\tau}{w_0}= h\ell^\omega$, $|h\ell| \leq (|\Pi|+1)\cdot |V|$ and $(\Cost_i(\outcome{\sigma}{w_0}))_{i\in \Pi} = (\Cost_i(\outcome{\tau}{w_0}))_{i\in\Pi}$.
	\end{itemize}
\end{proposition}

%

\subsection{Algorithms}
\label{section:algo}

In this section, we provide the main ideas behind our algorithms.

To solve \textbf{Problem~\ref{prob:decisionThreshold}}\footnote{As Problem~\ref{prob:decisionThreshold} is already solved in PSPACE for SPEs \cite{DBLP:journals/corr/abs-1905.00784} we here focus only on NEs.} (resp. \textbf{Problem~\ref{prob:decisionPareto}}) for NEs,  we use Corollary~\ref{cor:corCritOut2ConstraintProb} (resp. Proposition~\ref{prop:particularPareto}) which ensures that if there exists an NE which satisfies the conditions\footnote{Satisfying the conditions is either satisfying the constraints (Problem~\ref{prob:decisionThreshold} and Problem~\ref{prob:decisionWelfare}) or having a cost profile which is Pareto optimal (Problem~\ref{prob:decisionPareto}).}, there exists another one with a lasso outcome of polynomial length. The algorithm works as follows:\emph{(i)} it guesses a lasso of polynomial length;\emph{(ii)} it verifies that the cost profile of this lasso satisfies the conditions given by the problem (resp. is Pareto optimal in $\Plays(v_0)$) and \emph{(iii)} it verifies that the lasso is the outcome of an NE (Theorem~\ref{thm:critOutcomeNE}). Notice that this latter step is done in polynomial time as the lasso has a polynomial length and the values of the coalitional games are computed in polynomial time.

To solve \textbf{Problem~\ref{prob:decisionWelfare}} (resp. \textbf{Problem~\ref{prob:decisionPareto}} for SPEs), we use the algorithm designed for Problem~\ref{prob:decisionThreshold}. Each algorithm works as follows:\emph{(i)} it guesses a cost profile $c$;\emph{(ii)} it verifies that $c$ satisfies the conditions given by the problem and \emph{(iii)} it checks, thanks to the algorithm for Problem~\ref{prob:decisionThreshold}, if there exists an equilibrium with cost profile smaller than $c$ (resp. equal to $c$). 

Notice that for Problem~\ref{prob:decisionPareto}, we need to have an oracle allowing us to know if $c$ is Pareto optimal. This leads us to study Problem~\ref{prob:4} which lies in co-NP.
\begin{problem}
	\label{prob:4}
	Given a reachability game $(\mathcal{G},v_0)$ (resp. its extended game $(\mathcal{X},x_0)$) and a lasso $\rho \in \Plays(v_0)$ (resp. $\rho \in \Plays(x_0)$), we want to decide if $(\Cost_i(\rho))_{i\in\Pi}$ is  Pareto optimal  in $\Plays(v_0)$ (resp. $\Plays(x_0)$).
\end{problem}

\subsection{Results}
\label{subsection:results}

Thanks to the previous discussions in Section~\ref{section:algo}, we obtain the following results. Notice that we do not provide the proof for the NP-hardness (resp. PSPACE-hardness) as it is very similar to the one given in~\cite{condurache_et_al:LIPIcs:2016:6256} (resp. \cite{DBLP:journals/corr/abs-1905.00784}).
\begin{theorem}
	\label{thm:Complexity}
	Let $(\mathcal{G},v_0)$ be a quantitative reachability game.
	\begin{itemize}
		\item For NEs: Problem~\ref{prob:decisionThreshold} and Problem~\ref{prob:decisionWelfare} are NP-complete while Problem~\ref{prob:decisionPareto} is NP-hard and belongs to $\Sigma^P_2$.	
		\item For SPEs: Problems~\ref{prob:decisionThreshold},~\ref{prob:decisionWelfare} and~\ref{prob:decisionPareto} are PSPACE-complete.
	\end{itemize}
\end{theorem}

\begin{theorem}
	\label{thm:Memory}
		Let $(\mathcal{G},v_0)$ be a quantitative reachability game.
		\begin{itemize}
			\item For NEs: for each decision problem, if its answer is positive, then there exists a strategy profile $\sigma$ with memory in  $\mathcal{O}((|\Pi|+1)\cdot |V|)$ which satisfies the conditions.
			\item For SPEs: for each decision problem, if the answer is positive, then there exists a strategy profile $ \sigma$ with memory in  $\mathcal{O}(2^{|\Pi|} \cdot |\Pi|\cdot |V|^{(|\Pi|+1)\cdot(|\Pi|+|V|)+1})$  which satisfies the conditions.
			\item For both NEs and SPEs: \emph{(i)} for Problem~\ref{prob:decisionThreshold} and Problem~\ref{prob:decisionPareto}, $\sigma$ is such that: if $i \in \Visit(\outcome{\sigma}{v_0})$, $\Cost_i(\outcome{\sigma}{v_0}) \leq |\Pi|\cdot|V|$ and \emph{(ii)} for Problem~\ref{prob:decisionWelfare}, $\sigma$ is such that: $\sum_{i \in \Visit(\outcome{\sigma}{v_0})} \Cost_i(\outcome{\sigma}{v_0}) \leq |\Pi|^2 \cdot |V|$.
		\end{itemize}
\end{theorem}


%
%
%
 \bibliographystyle{splncs04}
 \bibliography{sv-biblio}

\newpage
\appendix

\section{Complements to Section~\ref{sec:existence}}
\label{app:existenceProblem}

\subsection{Proof of Theorem~\ref{thm:visitAll}}

To prove Theorem~\ref{thm:visitAll}, we begin with a preliminary lemma and the proof of Theorem~\ref{thm:visitAll} follows.

\begin{lemma} \label{lem:visitOne}
Let $\mathcal{G}$ be a quantitative reachability game. Then for all $v_0 \in V$ for which some target set $F_j$, $j \in \Pi$, is reachable from $v_0$, there exists an SPE in $(\mathcal{G},v_0)$ whose outcome $\rho$ visits at least one target set $F_i$, $i \in \Pi$, that is, $|\Visit(\rho)| \geq 1$.
\end{lemma}

\begin{proof}
By Theorem 2.1 in \cite{DBLP:journals/corr/abs-1205-6346}, there exists an SPE in $(\mathcal{G},v_0)$ for each initial vertex $v_0 \in V$. Consider the set $U \subseteq V$ of vertices $u$ for which some $F_j$ is reachable from $u$, and the set $U' \subseteq U$ of those vertices $u$ for which there is an SPE in $(\mathcal{G},u)$ that visits at least one target set. We have to prove that $U = U'$. 

Assume the contrary and let $v_0 \in U\setminus U'$. We claim that there exists an edge $(u,u')$ such that $u \in U \setminus U'$ and $u' \in U'$. Indeed as $v_0 \in U$, there exists a history $h = v_0v_1 \ldots v_k$ with $v_k \in F_j$ for some $j$. Hence $v_k \in U'$ since the outcome of all SPEs in $(\mathcal{G},v_k)$ immediately visits $F_j$. As along $h$ we begin with $v_0 \in U \setminus U'$ and we end with $v_k \in U'$, there must exist an edge $(v_\ell,v_{\ell + 1}) = (u,u')$ with $u \in U \setminus U'$ and $u' \in U'$. 

Let $\sigma^u$ (resp. $\sigma^{u'}$) be an SPE in $(\mathcal{G},u)$ (resp. in $(\mathcal{G},u')$). As $u' \in U'$, we can suppose that the outcome of $\sigma^{u'}$ visits some target set $F_j$. From $\sigma^{u}$ and $\sigma^{u'}$, we are going to construct another SPE $\tau$ in $(\mathcal{G},u)$ whose outcome will now visit this set $F_j$. This will lead to a contradiction with $u \in U\setminus U'$. We define such a strategy profile $\tau$ equal to $\sigma^u$ except that it is replaced by $\sigma^{u'}$ for all histories with prefix $uu'$. More precisely,
\begin{itemize}
\item for the particular history $u$, if $u \in V_i$, then $\tau_i(u) = u'$,  
\item for each history $uu'h \in \Hist_i$, $i \in \Pi$, we define $\tau_i(uu'h) = \sigma^{u'}_i(u'h)$,   
\item for each history $uv'h \in \Hist_i$, $i \in \Pi$, with $v' \neq u'$, we define $\tau_i(uv'h) = \sigma^{u}_i(uv'h)$.
\end{itemize}
Clearly the outcome of $\tau$ is equal to $u \outcome{\sigma^{u'}}{u'}$ and thus visits $F_j$. It remains to show that $\tau$ is an SPE, i.e., that $\rest{\tau}{h}$ is an NE in the subgame $(\rest{\mathcal{G}}{h},v)$ for all $hv \in \Hist_i(v)$, $i \in \Pi$.
\begin{itemize}
\item For all histories $hv$ that begin with $uv'$ with $v' \neq u'$, clearly $\rest{\tau}{h}$ is an NE in $(\rest{\mathcal{G}}{h},v)$ because $\rest{\tau}{h} = \rest{\sigma^u}{h}$ and $\sigma^u$ is an SPE. 
\item Take any history $hv$ that begin with $uu'$, and let $h = uh'$. Let $\tau'_i$ be a deviating strategy for player~$i$ in $(\rest{\mathcal{G}}{h},v)$. By definition of $\tau$ we have 
\begin{eqnarray*}
\outcome{\rest{\tau}{h}}{v} &=& u \outcome{\rest{\sigma^{u'}}{h'}}{v} \\
\outcome{(\tau'_i,\rest{\tau}{h,-i})}{v} &=& u \outcome{(\tau'_i,\rest{\sigma^{u'}}{h',-i})}{v}
\end{eqnarray*}
Moreover, as $u$ belongs to no target set, we have $\Cost_i(u \rho) = 1 + \Cost_i(\rho)$ for all plays $\rho \in \Plays(u')$. It follows that if $\tau'_i$ is a profitable deviation for player~$i$ with respect to $\rest{\tau}{h}$, it is also a profitable deviation with respect to $\rest{\sigma^{u'}}{h'}$. The latter case never holds because $\sigma^{u'}$ is an SPE (and in particular $\rest{\sigma^{u'}}{h'}$ is an NE). Therefore  $\rest{\tau}{h}$ is an NE in $(\rest{\mathcal{G}}{h},v)$.
\item It remains to consider the history $u$ and to prove that $\tau$ is an NE in $(\mathcal{G},u)$.  From what has been gathered so far, only player $i$ such that $u \in V_i$ might have a profitable deviation by deviating at the initial vertex $u$ with a strategy $\tau'_i$ such that $\tau'_i(u) = v' \neq u' = \tau_i(u)$. Notice that since $u \in U \setminus U'$, we have $\Cost_i(\outcome{\sigma^u}{u}) = + \infty$ and since $\sigma^u$ is an SPE (and in particular an NE), we have $\Cost_i(\outcome{\tau'_i, \sigma^u_{-i}}{u}) = + \infty$. Moreover as $\tau'_i(u) = v' \neq u'$ and by definition of $\tau$, we have $\Cost_i(\outcome{\tau'_i, \sigma^u_{-i}}{u}) = \Cost_i(\outcome{\tau'_i, \tau_{-i}}{u}) = + \infty$. It follows that $\tau'_i$ is not a profitable deviation for player~$i$ with respect to $\tau$, and then $\tau$ is an NE in $(\mathcal{G},u)$.
\end{itemize}
\qed\end{proof}

\begin{proof}[of Theorem~\ref{thm:visitAll}] 
Let $(\mathcal{G},v_0)$, with $\mathcal{G} = (\mathcal{A}, (\Cost_i)_{i\in \Pi}, (F_i)_{i\in \Pi})$, be an initialized quantitative reachability game such that its arena is strongly connected. Assume by contradiction that there exists no SPE in $(\mathcal{G},v_0)$
whose outcome visits all target sets $F_i$, $i \in \Pi$, that are non empty. By Theorem 2.1 in \cite{DBLP:journals/corr/abs-1205-6346}, there exists an SPE $\sigma$ in $(\mathcal{G},v_0)$, and we take such an SPE $\sigma$ whose outcome $\rho = \outcome{\sigma}{v_0}$ visits a maximum number of target sets, say $F_{i_1}, F_{i_2}, \ldots, F_{i_k}$. Thus by assumption there exists at least one $F_j \neq \emptyset$ with $j \not\in \{i_1, \ldots, i_k\}$ that is not visited by $\rho$. Thanks to Lemma~\ref{lem:visitOne}, we are going to define from $\sigma$ another SPE $\tau$ in $(\mathcal{G},v_0)$ whose outcome visits all $F_{i_1}, \ldots, F_{i_k}$ as well as an additional target set. This will lead to a contradiction.

Consider a prefix $\rho_0\rho_1 \ldots \rho_\ell$ of $\rho$ that visits all $F_{i_1}, \ldots, F_{i_k}$. We denote it by $gu$ with $u = \rho_\ell$. From $\mathcal G$ we define the quantitative reachability game $\mathcal{G}' = (\mathcal{A}, (\Cost'_i)_{i\in \Pi}, (F'_i)_{i\in \Pi})$ with the same arena $\arena$ and such that $F'_i = \emptyset$ if $i \in \{i_1, \ldots, i_k\}$ and $F'_i = F_i$ otherwise ($(\Cost'_i)_{i\in \Pi}$ is defined with respect to $(F'_i)_{i\in \Pi}$ as in Definition~\ref{def:quantitativeGame}). Notice that $F'_j = F_j$ is not empty and it is reachable from $u$ since $\arena$ is strongly connected. Therefore by Lemma~\ref{lem:visitOne}, there exists an SPE $\sigma'$ in $(\mathcal{G}',u)$ that visits at least one target set $F'_{j'}$. From $\sigma$ and $\sigma'$, we define a strategy profile $\tau$ in $(\mathcal{G},v_0)$ as follows: let $h \in \Hist_i(v_0)$,
\begin{itemize}
\item if $h = guh'$ for some $h'$, then $\tau_i(h) = \sigma'_i(u h')$,
\item otherwise $\tau_i(h) = \sigma_i(h)$.
\end{itemize}
Thus, $\tau$ acts as $\sigma$, except that after a history beginning with $gu$, it acts as $\sigma'$. Clearly the outcome of $\tau$ is equal to $g \outcome{\sigma'}{u}$ and thus visits $F'_{j'} = F_{j'}$ in addition to $F_{i_1}, \ldots, F_{i_k}$. It remains to show that $\tau$ is an SPE. 
Consider $hv \in \Hist_i(v_0)$, $i \in \Pi$, and let us show that $\rest{\tau}{h}$ is an NE in $(\rest{\mathcal{G}}{h},v)$.
\begin{itemize}
\item If neither $hv$ is a prefix of $gu$ nor $gu$ is a prefix of $hv$, then $\rest{\tau}{h} = \rest{\sigma}{h}$ by definition of $\tau$, and $\rest{\tau}{h}$ is an NE in $(\rest{\mathcal{G}}{h},v)$ because $\sigma$ is an SPE in $(\mathcal{G},v_0)$.
\item If $gu$ is a prefix of $hv$, let $h'$ such that $gh' = h$. Suppose first that $hv$ visits $F_i$, then player~$i$ has clearly no incentive to deviate in $(\rest{\mathcal{G}}{h},v)$. Suppose now that $hv$ does not visit $F_i$, then $i \not\in \{i_1,\ldots,i_k\}$ and $F'_i = F_i$ by definition of $\mathcal{G}'$. Hence for all plays $\pi$ in $(\rest{\mathcal{G}}{h},v)$ that start in $v$, $h'\pi$ is a play in $(\mathcal{G'},u)$ that start in $u$, and we have $\Cost_i(h\pi) = |gu| + \Cost'_i(h'\pi)$. Hence by definition of $\tau$, a profitable deviation for player~$i$ with respect to $\rest{\tau}{h}$ $(\rest{\mathcal{G}}{h},v)$ would be a profitable deviation with respect to $\rest{\sigma'}{h'}$ in $(\rest{\mathcal{G'}}{h'},v)$. The latter case cannot happen as $\sigma'$ is an SPE in $(\mathcal{G}',u)$ and it follows that $\rest{\tau}{h}$ is an NE in $(\rest{\mathcal{G}}{h},v)$.
\item Consider the last case where $hv$ is a prefix of $gu$ with $hv \neq gu$, and let $hh' = g$. Consider $\tau'_i $ a deviating strategy for player~$i$ with respect to $\rest{\tau}{h}$ in the subgame $(\rest{\mathcal{G}}{h},v)$, and let $\rho' = \outcome{(\tau'_i, \rest{\tau}{h,-i})}{v}$. Without loss of generality, we can suppose that $h'u$ is not a prefix of $\rho'$ since this case was treated at the previous item. 
Notice that if $i \in \{i_1,\ldots,i_k\}$, then $\Cost_i(\outcome{\rest{\tau}{h}}{v}) = \Cost_i(\outcome{\rest{\sigma}{h}}{v})$, otherwise $
\Cost_i(\outcome{\rest{\tau}{h}}{v}) \leq + \infty = \Cost_i(\outcome{\rest{\sigma}{h}}{v})$. In both cases, as $h'u$ is a prefix of both $\outcome{\rest{\tau}{h}}{v}$ and $\outcome{\rest{\sigma}{h}}{v}$, but not a prefix of $\rho'$, if $\tau'_i$ was a profitable deviation for player~$i$ with respect to $\rest{\tau}{h}$, it would also be a profitable deviation with respect to $\rest{\sigma}{h}$ which is impossible since $\sigma$ is an SPE. 
\end{itemize}
\qed\end{proof}

We end with an example which shows that the hypothesis Theorem~\ref{thm:visitAll} that the arena is strongly connected is necessary.

\begin{example} \label{ex:0}
Consider the initialized qualitative reachability game $(\mathcal{G},v_0)$ of Figure~\ref{fig:gameNotStrongly}. There are two players, Player 1 who owns round vertices and Player 2 who owns square vertices, and $F_1 = \{v_1\}$, $F_2 = \{v_2\}$. Clearly there is a unique NE $\sigma = (\sigma_1,\sigma_2)$ in $(\mathcal{G},v_0)$ such that $\sigma_1(v_0) = v_1$ and $\sigma_2(v_1) = v_1$, $\sigma_2(v_2) = v_2$. Its outcome only visits  $F_1$ (and not $F_2$).
	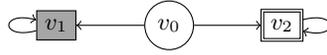
\begin{figure}[ht!]
		\centering
		\begin{tikzpicture}
			\node[draw,circle] (v0) at (0,0){$v_0$};
			\node[draw,fill=gray!70] (v1) at (-1.5,0){$v_1$};
			\node[draw, accepting] (v2) at (1.5,0){$v_2$};
			
			\draw[->] (v0) -- (v1);
			\draw[->] (v0) -- (v2);
			\draw[->] (v1) to [loop left] (v1);
			\draw[->] (v2) to [loop right] (v2);
			
		\end{tikzpicture}
		\caption{A qualitative reachability game with two players such that $F_1 = \{v_1\}$ and $F_2 = \{v_2\}$.}
		\label{fig:gameNotStrongly}
	\end{figure}
\qed\end{example}


\section{Complements to Section~\ref{subsection:charac}}

\subsection{Complements about coalitionnal games}

We provide the formal definitions of (quantitative) coalitional game, value and optimal strategy.

\begin{definition}[(Quantitative) Coalitional game]
	\label{def:coalitionalGame}
	
	 Let $\arena = (\Pi, V, E; (V_i)_{i\in\Pi})$ be an arena and $\mathcal{G}=(\arena, (\Cost_i)_{i\in\Pi}, (F_i)_{i\in\Pi})$ be a quantitative reachability game with $|\Pi| \geq 2$. With each player $i\in \Pi$, we can associate a two-player zero-sum quantitative reachability game depicted by $\mathcal{G}_i = (\arena_i, (\CostMin, \GainMax), F)$ and defined as follows: \emph{i)} $\arena_i = (\{i,-i\}, V, (V_i,V\backslash V_i), E)$ where Player $i$ (resp. $-i$) can be called Player Min (resp. Player Max); \emph{ii)} $\CostMin = \Cost_i$ and $\GainMax = \CostMin$ and \emph{iii)} $F = F_i$.
	%
\end{definition}

\begin{definition}[Value]
	\label{def:value}
	 Let $\mathcal{G}_i$ be a coalitional game and $v\in V$ be a vertex, we define the value of $\mathcal{G}_i$ from $v$ as :
	\begin{equation} \Val_i(v) = \inf_{\sigma_{1} \in \Sigma_{Min}} \sup_{\sigma_2 \in \Sigma_{Max}} \CostMin(\outcome{\sigma_1,\sigma_2}{v}). \label{eq:defValue} \end{equation} \end{definition}

Remark that, as for each $i \in \Pi$ the coalitional game $\mathcal{G}_i$ is determined (\cite{DBLP:conf/lfcs/BrihayePS13}) and $\CostMin = \GainMax$, the equality~\eqref{eq:defValue} could be defined as $Val_i(v)= \sup_{\sigma_2 \in \Sigma_{Max}} \inf_{\sigma_{1} \in \Sigma_{Min}} \GainMax(\outcome{\sigma_1,\sigma_2}{v})$.

\begin{definition}[Optimal strategy]
	\label{def:optimalStrat}
	Let $\mathcal{G}_i$ be a coalitional game, we say that $\sigma_1^* \in \Sigma_{Min}$ is an \emph{optimal strategy} for player $Min$ if, for all $v \in V$,
	 we have that:
	$ \sup_{\sigma_2 \in \Sigma_{Max}} \CostMin( \outcome{\sigma_1^*,\sigma_2}{v}) \leq \Val_i(v) $. Similarly, we say that $\sigma_2^* \in \Sigma_{Max}$ is an optimal strategy for player $Max$ if, for all $v \in V$,
     we have that:	

	$\inf_{\sigma_1 \in \Sigma_{Min}} \GainMax (\outcome{\sigma_1, \sigma_2^*}{v}) \geq \Val_i(v)$.	
	
\end{definition}

\subsection{Proof of Theorem~\ref{thm:critOutcomeNE}}
\label{app:critereOutNE}

\begin{proof}[of Theorem~\ref{thm:critOutcomeNE}]
Let us first recall that, for all $i\in \Pi$, the coalitional game $\mathcal{G}_i$ is determined and there are optimal positional strategies for both players $(\sigma^*_i, \sigma^*_{-i})$. Moreover as $\CostMin = \GainMax$, for all $v \in V$, we have: 
$$ \inf_{\sigma_i \in \Sigma_{Min}} \CostMin(\outcome{\sigma_i, \sigma^*_{-i}}{v}) = \Val_i(v) = \sup_{\sigma_{-i} \in \Sigma_{Max}} \CostMin(\outcome{\sigma^*_i, \sigma_{-i}}{v}).$$
From the optimal strategy $\sigma^*_{-i}$ in $\mathcal{G}_i$ we can extract a strategy $\sigma^*_{j,i}$ in $\mathcal{G}$. Notice also that even if $\sigma^*_i$ is a strategy in $\mathcal{G}_i$, we can use it as a strategy for Player $i$ in $\mathcal{G}$.
Let us prove the equivalence between the two assertions.\\

$\mathbf{1\Rightarrow 2}$:
	Let $\sigma$ be a Nash equilibrium in $(\mathcal{G},v_0)$ such that $\outcome{\sigma}{v_0} = \rho$.  Let us assume by contradiction that there exist $i \in \Pi$ and $k \in \mathbb{N}$ such that $i \not \in \Visit(\rho_0\ldots\rho_k)$ and $\rho_k \in V_i$ such that:
	\begin{equation}
		\Cost_i(\rho_{\geq k}) > \Val_i(\rho_k). \label{eq:critContrad1}
	\end{equation}
	Let $h = \rho_0\ldots \rho_{k-1}$, we can write:
	\begin{equation}
		\Cost_i(\rho_{\geq k}) = \Cost_i(\outcome{\rest{\sigma}{h}}{\rho_k}). \label{eq:crit1}
	\end{equation}
	Additionally, by definition of value in a coalitional game and thanks to the fact that the optimal strategies are positional:
	\begin{align}
		\Val_i(\rho_k) &= \sup_{\tau_{-i} \in \Sigma_{Max}} \CostMin(\outcome{\sigma^*_i,\tau_{-i}}{\rho_k}) \notag\\
					   &\geq \CostMin(\outcome{\sigma^*_i, \sigma_{-i \restriction h}}{\rho_k}) \notag\\
					   &= \Cost_i(\outcome{\sigma^*_i, \sigma_{-i \restriction h}}{\rho_k}) \label{eq:crit2}					
	\end{align} 
	
	where $\sigma^*_i$ is the optimal strategy of Player $i$ in $\mathcal{G}_i$ and $\sigma_i$ is an abuse of notation to depict the strategy of the coalition $-i = \Pi\backslash\{i\}$ which follows strategies $\sigma_j$ for all  $j \neq i$.
	
	By~\eqref{eq:critContrad1},~\eqref{eq:crit1} and~\eqref{eq:crit2}, it follows that:
	
	$$ \Cost_i(\outcome{\sigma^*_i, \sigma_{-i \restriction h}}{\rho_k}) < \Cost_i(\outcome{\rest{\sigma}{h}}{\rho_k}).$$
	
	As $i \not \in \Visit(h)$ by hypothesis, we can conclude that:
	
	$$\Cost_i(h\outcome{\sigma^*_i, \sigma_{-i \restriction h}}{\rho_k}) < \Cost_i(h\outcome{\rest{\sigma}{h}}{\rho_k}) = \Cost_i(\rho). $$
	
	This means that following $\sigma_i$ along $h$ and then $\sigma^*_i$ once he reaches $\rho_k$ is a profitable deviation for Player i. This concludes the proof.\\
	
$\mathbf{2\Rightarrow 1}$: Let $\tau$ be a strategy profile such that $\outcome{\tau}{v_0} = \rho$. From $\tau$ we aims to construct a Nash equilibrium with the same outcome. The main idea is the following one: first, all player play according to $\tau$. But if a player, let us call him Player $i$ deviates from $\tau_i$, the other players form a coalition and each of them plays their strategy obtained thanks to the strategy $\sigma^*_{-i}$ in $\mathcal{G}_i$.

In order to define properly the researched Nash equilibrium, we have to define  a punishment function $P : \Hist(v_0) \rightarrow \Pi \cup \{ \perp \}$ which allows us to know who is the  player who has deviated for the first time from $\tau$. So for all $h \in \Hist(v_0)$, $P(h) = \perp$ if no player has yet deviated and $P(h)= i$ for some $i \in \Pi$ if Player $i$ is the first player who has deviated along $h$. We can define $P$ as follows: for the initial vertex $P(v_0) = \perp$ and then for all history $hv \in Hist(v_0)$ with $v\in V$:

$$ P(hv) = \begin{cases} \perp & \text{if } P(h) = \perp \text{ and } hv \text{ is a prefix of } \rho, \\
                          i & \text{if } P(h) = \perp \text{, } hv \text{ is not a prefix of } \rho \text{ and } h \in \Hist_i,\\
						  P(h) & \text{ otherwise.}\end{cases}.$$
						
We now define $\sigma$. For all $i \in \Pi$ and for all $h \in \Hist_i(v_0)$:

$$ \sigma_i(h) = \begin{cases} \tau_i(h) & \text{if } P(h) = \perp, \\
							   \sigma^*_i(h) & \text{if } P(h)=i, \\
                               \sigma^*_{i,P(h)}(h) & \text{ otherwise} \end{cases}.$$

It is clear that $\outcome{\sigma}{v_0} = \rho$. It remains to prove that $\sigma$ is a Nash equilibrium in $(\mathcal{G},v_0)$. Let us assume that $\sigma$ is not an NE.  It means that there exists a profitable deviation depicted by $\tilde{\sigma}_i$ for some player $i$. We chose $i$ such that $i$ is the first player who has profitable deviation from $\sigma$ along $\rho$. Let $\tilde{\rho} = \outcome{\tilde{\sigma}_i, \sigma_{-i}}{v_0}$ the outcome such that Player $i$ plays his profitable deviation. As $\tilde{\sigma}_i$ is a profitable deviation we have:

\begin{equation} \Cost_i(\tilde{\rho}) < \Cost_i(\rho). \label{eq:crit3}\end{equation}
	
	Moreover as $\rho$ and $\tilde{\rho}$ both begin in $v_0$, they have a common prefix. Let $hv\in \Hist_i$ this longest common prefix. 
	We have that: $\rho = h \outcome{\rest{\sigma}{h}}{v}$ and $\tilde{\rho} = h \outcome{\rest{\tilde{\sigma_{i}}}{h}, \sigma_{-i \restriction h}}{v}$. Notice that $i \not \in \Visit(hv)$. 
	But, by definition of $\sigma$ and as the optimal strategies in $\mathcal{G}_i$ are positional, we can rewrite these two equalities as follows: $ \rho = h \outcome{\tau_{\restriction h}}{v}$ and $\tilde{\rho} = h \outcome{\rest{\tilde{\sigma_{i}}}{h}, (\sigma^*_{j,i})_{j \in \Pi\backslash\{i\}}}{v}.$
	Additionally, thanks to the definition of the value in the coalitional game $\mathcal{G}_i$:
	
	\begin{align}
		\Val_i(v) &= \inf_{\mu_i \in \Sigma_{Min}} \CostMin(\outcome{\mu_i,\sigma^*_{-i}}{v}) \notag \\
                  &\leq \CostMin(\outcome{\tilde{\sigma}_{i\restriction h}, \sigma^*_{-i}}{v}) \notag \\	
                  &=\Cost_i(\outcome{\rest{\tilde{\sigma_{i}}}{h}, (\sigma^*_{j,i})_{j \in \Pi\backslash\{i\}}}{v}).\label{eq:crit4}
	\end{align}
	By hypothesis, as $hv$ is a prefix of $\rho$ and $i \not \in \Visit(hv)$, we have that $\Val_i(v) \geq \Cost_i(\outcome{\tau_{\restriction h}}{v})$. Thus by~\eqref{eq:crit4}, it follows that:
	\begin{equation*}  \Cost_i(\outcome{\rest{\tilde{\sigma_{i}}}{h}, (\sigma^*_{j,i})_{j \in \Pi\backslash\{i\}}}{v}) \geq \Cost_i(\outcome{\tau_{\restriction h}}{v}). \end{equation*}
		
		And thanks to the definition of the cost function associated with quantitative reachability games, we have that:
		\begin{equation*}  \Cost_i(h\outcome{\rest{\tilde{\sigma_{i}}}{h}, (\sigma^*_{j,i})_{j \in \Pi\backslash\{i\}}}{v}) \geq \Cost_i(h\outcome{\tau_{\restriction h}}{v}). \end{equation*}
			
			Thus, we can conclude that $\Cost_i(\tilde{\rho}) \geq \Cost_i(\rho)$ which leads to a contradiction with~\eqref{eq:crit3}. This concludes the proof. \qed
\end{proof}

\subsection{Finite machine which represents strategies in Theorem~\ref{thm:critOutcomeNE}}
\label{app:memoryCritNE}

A finite-state Machine ${\cal M} = (M, m_0, \alpha_u, \alpha_n)$ is such that $M$ is a finite set of states (the memory of the strategy), $m_0 \in M$ is the initial memory state, $\alpha_u\colon M \times V \rightarrow M$ is the update function, and $\alpha_n\colon M \times V_i \rightarrow V$ is the next-action function. The machine $\cal M$ defines a strategy $\sigma_i$ such that $\sigma_i(h v) = \alpha_n(\widehat{\alpha}_u(m_0,h),v)$ for all histories $h v \in \Hist_i$, where $\widehat{\alpha}_u(m,\epsilon) =  m$ and  $\widehat{\alpha}_u(m,hv) = \alpha_u(\widehat{\alpha}_u(m,h),v)$ for all $m \in M$ and $hv \in \Hist$. The \emph{size} of the strategy $\sigma_i$ is the size $|M|$ of its machine $\cal M$. Note that $\sigma_i$ is positional  when $|M| = 1$.\\

 Formally, let $\rho = \rho_0\ldots \rho_{k-1}(\rho_k\ldots\rho_n)^\omega$, we define for all $i \in \Pi$ a finite-state machine $\mathcal{M}_i=(M,m_0,\alpha_u,\alpha_n)$ where:
\begin{itemize}
	\item $M = \{ \rho_0\rho_0, \rho_0\rho_1, \ldots,\rho_{n-1}\rho_n, \rho_n\rho_k\} \cup \Pi$.\\
	The set $\{ \rho_0\rho_0, \rho_0\rho_1, \ldots,\rho_{n-1}\rho_n, \rho_n\rho_k\}$ allows us to be sure that the outcome $\rho$ is well followed by all the players.  Once it is no longer the case, we only have to retain who has deviated, this is the role of $\Pi$. Notice that, even if we can have $\rho_m\rho_{m+1} = \rho_{m'}\rho_{m'+1}$ along $\rho$, the edges $\rho_m\rho_{m+1}$ and $\rho_{m'}\rho_{m'+1}$ are depicted by two different memory state in $M$. Thus $|M| = |h\ell|+2+|\Pi|$.
	\item $m_0= \rho_0\rho_0$ is the memory state which specifies that the plays has not begun yet.
	\item $ \alpha_u : M \times V \rightarrow M$ is defined as follows: for all $m \in M$ and $v \in V$:
	$$ \alpha_u(m,v) = \begin{cases} j & \text{if } m=j\in \Pi \text{ or } (m=v_1v_2, \text{ with } v_1,v_2\in V, v \neq v_2 \text{ and } v_1 \in V_j)\\
	  								\rho_t\rho_{t+1} & \text{if } m = u\rho_t (\text{with } t\in\{0,\ldots,|h\ell|-1\}), u \in V \text{ and } v= \rho_t \\
									\rho_n\rho_k & \text{ otherwise } (m=u\rho_n \text{ and } v= \rho_n).\end{cases}.$$
	\item $ \alpha_n: M \times V_i \rightarrow V$ is defined in this way: for all $m \in M$ and $v \in V_i$:
	$$ \alpha_n(m,v) = \begin{cases} \rho_1 & \text{if } m =\rho_0\rho_0  \text{ and } v = \rho_0 \\ 
									 \rho_{t+2}& \text{if } m = \rho_{t}\rho_{t+1} (\text{with } t \in \{0,\ldots, |h\ell|-2 \}) \text{ and }v= \rho_{t+1} \\
									\rho_k & \text{if } m = \rho_{n-1}\rho_n \text{ and } v= \rho_n \\
									\sigma^*_i(v) & \text{if } m=i \\
									\sigma^*_{i,j}(v) & \text{ otherwise}\end{cases}.$$
\end{itemize}

\subsection{Complements about extended game}

We here provide the formal definition of an extended game.

\begin{definition}[Extended game] \label{def:extGame}
Let $\mathcal{G} = (\mathcal{A}, (\Cost_i)_{i\in \Pi},  (F_i)_{i\in\Pi})$ be a  quantitative reachability game with an arena $\mathcal{A} = (\Pi, V, E, (V_i)_{i\in \Pi})$, and let $v_0$ be an initial vertex. The \emph{extended game} of $\mathcal{G}$ is equal to $\extGame = (\extG, (\Cost^X_i)_{i\in \Pi}, (\extFi{i})_{i\in\Pi})$ with the arena $\extG = (\Pi, \extV, \extE, (\extVi{i})_{i\in\Pi})$, such that:
\begin{itemize}
    \item $\extV = V \times 2^\Pi$
    \item $((v,I),(v',I')) \in \extE$ if and only if $(v,v')\in E$ and $I' = I \cup \{i \in \Pi \mid v' \in F_i \}$
    \item $(v,I) \in \extVi{i}$ if and only if $v\in V_i$
    \item $(v,I) \in \extFi{i}$ if and only if $i \in I$
    \item for each $\rho \in \Plays_X$, $\Cost^X_i(\rho)$ is equal to the least index $k$ such that $\rho_k \in \extFi{i}$, and to $+\infty$ if no such index exists.
\end{itemize}
The initialized extended game $(\extGame,x_0)$ associated with the initialized game $({\mathcal G}, v_0)$ is such that $x_0 = (v_0,I_0)$ with $I_0 = \{i \in \Pi \mid v_0 \in F_i \}$. 
\end{definition}

\section{Complements to Section~\ref{section:lassoes}}

\begin{lemma}
	\label{lem:Lassoes}
	Let $(\mathcal{G},v_0)$ be a quantitative reachability game and $\rho \in \Plays$ be a play.
	\begin{itemize}
		\item If $\rho'$ is  obtained by applying (P1) on $\rho$, then $(\Cost_i(\rho'))_{i \in \Pi} \leq (\Cost_i(\rho))_{i\in\Pi} $
		\item If $\rho'$ is  obtained by applying (P2) on $\rho$, then $(\Cost_i(\rho'))_{i\in\Pi} = (\Cost_i(\rho))_{i\in\Pi}$.	
		\item Applying (P1) until it is no longer possible and then (P2), leads to a lasso $\rho'$ with length at most $(|\Pi|+1)\cdot |V|$ and $\Cost_i(\rho')\leq |V|\cdot|\Pi|$ for each $i \in \Visit(\rho')$.			
	\end{itemize}
\end{lemma}

\begin{remark}[about Lemma~\ref{lem:Lassoes}]
	Notice that, given a quantitative reachability game $(\mathcal{G},v_0)$, as its extended game $(\mathcal{X},x_0)$ is in  particular a quantitative reachability game, all statements of Lemma~\ref{lem:Lassoes} hold. 
	
	 But, if we only apply the third assertion on $(\mathcal{X},x_0)$, we obtain bounds on the size of the lasso and on the cost of plays which depends on $|V^X|$. It means that it is exponential on the size of the initial game $\mathcal{G}$. 
	
	In fact, even for the extended game $(\mathcal{X},x_0)$ we can obtain the that: applying (P1) until it is no longer possible and then (P2), leads to a lasso $\rho'$ with size at most $(|\Pi|+1)\cdot |V|$ and $\Cost_i(\rho')\leq |V|\cdot|\Pi|$ for each $i \in \Visit(\rho')$ where $|V|$ is the number of vertices in $\mathcal{G}$.
\end{remark}

\subsection{Proof of Lemma~\ref{lem:LassoesConsis}}

\begin{proof}[of Lemma~\ref{lem:LassoesConsis}]  We begin by proving the assertion for (P1). Let $\rho$ be a $\lambda$-consistent play and we apply (P1) on it to obtain $\rho'$. Then, for all $i \in \Pi$, $\Cost_i(\rho') \leq \Cost_i(\rho)$. And in particular, for all $i\in \Pi$ and for all $n \in \mathbb{N}$, \begin{equation} \Cost_i(\rho'_{\geq n}) \leq \Cost_i(\rho_{\geq \varphi(n)}) \label{eq:corCostUnCycle}\end{equation} where $\varphi$ is the injective function which matches a node in $\rho'$ with its corresponding node in $\rho$, \emph{i.e.,} for all $n \in \mathbb{N}$, $\rho'_n = \rho_{\varphi(n)}$.
		Let $i \in \Pi$ and  $k \in \mathbb{N}$ such that $i \not \in \Visit(\rho'_0\ldots \rho'_k)$ and $\rho'_k \in V_i$, we have to prove that $\Cost_i(\rho'_{\geq k}) \leq \lambda(\rho'_k)$. By construction of $\rho'$  and by~\eqref{eq:corCostUnCycle}, we have that $$\Cost_i(\rho'_{\geq k}) \leq \Cost_i(\rho_{\geq \varphi(k)}).$$

		But, as $\rho$ is $\lambda$-consistent, we also have that:

		$$ \Cost_i(\rho_{\geq \varphi(k)}) \leq \lambda(\rho_{ \varphi(k)})
		= \lambda_i( \rho'_k) $$

		and we can conclude that $\Cost_i(\rho'_{\geq k}) \leq \lambda(\rho'_k)$ which proves that $\rho'$ is $\lambda$-consistent.		This concludes the proof.\\
		
	For (P2), the assertion holds because $\rho'$ is a copy of $\rho$ until each player in $\Visit(\rho)$ has visited his target set and $\Visit(\rho) = \Visit(\rho')$.
		\qed	
\end{proof}


\subsection{Proof of Corollary~\ref{cor:corCritOut2ConstraintProb} }

\subsubsection{Proof of Corollary~\ref{cor:corCritOut2ConstraintProb} for NEs}
\label{app:corCritOut1}

The proof of Corollary~\ref{cor:corCritOut2ConstraintProb} for NEs is due to Corollary~\ref{cor:corCritOut1} that we prove below. 

\begin{corollary}
	\label{cor:corCritOut1}
	Let $\sigma$ be an NE in a quantitative reachability game $(\mathcal{G},v_0)$, then there exists $\tau$ an NE in $(\mathcal{G},v_0)$ such that:
	\begin{itemize} \item  $\outcome{\tau}{v_0}$ is a lasso $h\ell^\omega$ such that $|h\ell|\leq (|\Pi|+1)\cdot |V| $; 
		\item for all $i \in \Visit(\outcome{\tau}{v_0})$, $\Cost_i(\outcome{\tau}{v_0}) \leq  \min \{ \Cost_i(\outcome{\sigma}{v_0}), |\Pi|\cdot |V| \}$ and for all $i \not \in \Visit(\outcome{\tau}{v_0})$, $\Cost_i(\outcome{\tau}{v_0}) = \Cost_i(\outcome{\sigma}{v_0})= +\infty$;
		\item  the memory of $\tau$ is in $\mathcal{O}((|\Pi|+1)\cdot |V|)$. 
	\end{itemize}
\end{corollary}

\begin{proof}[of Corollary~\ref{cor:corCritOut1}]
	
	Let $\rho$ be a play in $(\mathcal{G},v_0)$ such that $ \rho = \outcome{\sigma}{v_0}$. We apply procedure (P1) on $\rho$ until there is no longer an unnecessary cycle and then we apply (P2). In this way, we obtain a lasso $\rho' = h \ell^\omega \in \Plays(v_0)$. By Lemma~\ref{lem:Lassoes}, $|h\ell| \leq (|\Pi|+1)\cdot |V|$ and  $\Cost_i(h\ell^\omega) \leq \min \{\Cost_i(\outcome{\sigma}{v_0}), |\Pi|\cdot|V| \}$ if $i \in \Visit(\outcome{\sigma}{v_0})$ and $\Cost_i(h\ell^\omega) = +\infty$ otherwise.

 By hypothesis and thanks to Theorem~\ref{thm:critOutcomeNE}, we know that $\rho$ is $\ConstNE$-consistent. Thus, by Lemma~\ref{lem:LassoesConsis}, $\rho'$ is $\ConstNE$-consistent. And Theorem~\ref{thm:critOutcomeNE} for $\rho'$ concludes the proof.  \qed
\end{proof}

\begin{proof}[of Corollary~\ref{cor:corCritOut2ConstraintProb} for NEs] It is a direct consequence of Corollary~\ref{cor:corCritOut1}. \qed
\end{proof}

\subsubsection{Proof of Corollary~\ref{cor:corCritOut2ConstraintProb} for SPEs}
In this section we assume that all the definitions and notations introduced in~\cite{DBLP:journals/corr/abs-1905.00784} are known.

By adapting the concept of \emph{(good) symbolic witness} (a set of lassoes with some good properties) used in~\cite{DBLP:journals/corr/abs-1809-03888}, we can show that if there exists an SPE with a cost profile $c$ then, there exists one with the same cost profile but with a finite-memory. This leads to Proposition~\ref{prop:SPEmemory} which allows us to prove  Corollary~\ref{cor:corCritOut2ConstraintProb} for SPEs.

Before the statement of the proposition, we formally introduce what is a (good) symbolic witness. This notion was introduced in~\cite{DBLP:journals/corr/abs-1809-03888} for games with prefix-independent gain functions. We adapt it for quantitative reachability games for which the cost function is not prefix-independent.

\begin{definition}[Symbolic witness]
Let $(\mathcal{G},v_0)$ be an initialized quantitative reachability game and $(\mathcal{X}, (v_0,I_0))$ its extended game. Let $\mathcal{I}$ be a subset of $(\Pi \cup \{0\}) \times V \times 2^\Pi$ such that:
\begin{align*} \mathcal{I} = \{(0,v_0,I_0)\} &~\cup~ 
\{(i,v',I') \mid \mbox{there exists }((v,I),(v',I')) \in E^X \\ &\mbox{ with } (v,I),(v',I') \in \Succ^*(v_0,I_0) \mbox{ and } v \in V_i \}. \end{align*}

A \emph{symbolic witness} is a set $\mathcal{P} = \{ \rho_{i,v,I} \mid (i,v,I) \in \mathcal{I}\}$ such that each $\rho_{i,v,I}$ is a lasso in $\mathcal{X}$ with $\First(\rho_{i,v,I}) = (v,I)$.
\end{definition}

\begin{definition}[Good symbolic witness]
\label{def:Good}
A symbolic witness $\mathcal{P}$ is \emph{good} if for all $\rho_{j,u,J}$, $\rho_{i,v',I'} \in \mathcal{P}$, for all suffix $\rho \in \Plays(v,I)$ of $\rho_{j,u,J}$ such that $((v,I),(v',I'))\in E^X$ and $(v,I)\in V_i^X$, if $i \not \in I$, then we have: $$ \Cost_i(\rho) \leq 1 + \Cost_i(\rho_{i,v'I'}).$$
 \end{definition}

\begin{figure}[h!]
	\centering
	\begin{tikzpicture}
		\node[draw,rounded corners=4pt,minimum width=25pt, minimum height=15pt](u) at (0,0){$(u,J)$};
		\node(point1) at (1.5,0){$\ldots$};
		\node[draw,rounded corners=4pt,minimum width=25pt, minimum height=15pt](v) at (3,0){$(v,I)$};
		\node(vi) at (3,-0.7){$\in V_i$};
		\node[draw,rounded corners=4pt,minimum width=25pt, minimum height=15pt](v') at (3.5,1){$(v',I')$};
		\node(point2) at (5,1){$\ldots$};
		\node[draw,rounded corners=4pt,minimum width=25pt, minimum height=15pt](blanc1) at (6.5,1){}; 
		\node(point3) at (8,1){$\ldots$};
		\node[draw,rounded corners=4pt,minimum width=25pt, minimum height=15pt](blanc2) at (9.5,1){}; 
		\node (path1) at (10.5,1){$\rho_{i,v',I'}$};
		\node(point4) at (5,0){$\ldots$};
		\node[draw,rounded corners=4pt,minimum width=25pt, minimum height=15pt](blanc3) at (6.5,0){}; 
		\node (point5) at (8,0){$\ldots$};
		\node[draw,rounded corners=4pt,minimum width=25pt, minimum height=15pt](blanc4) at (9.5,0){};
		\node (path2) at (10.5,0){$\rho_{j,u,J}$};	
		
		\draw [decorate,decoration={brace,mirror,amplitude=10pt},xshift=-4pt,yshift=0pt]
(2.7,-1) -- (10,-1) node [black,below=0.5cm,midway]{$\rho$} ;

		\draw[->] (u) -- (1,0);
		\draw[->] (2,0) -- (v);
		\draw[->] (v) -- (v');
		\draw[->] (v') -- (4.5,1);
		\draw[->] (5.5,1) --(blanc1);
		\draw[->] (blanc1) -- (7.5,1);
		\draw[->] (8.5,1) -- (blanc2);
		\draw[->] (blanc2) to [bend right] (blanc1);
		
		\draw[->] (v) -- (4.5,0);
		\draw[->] (5.5,0) -- (blanc3);
		\draw [->] (blanc3) -- (7.5,0);
		\draw[->] (8.5,0) -- (blanc4);
		\draw[->] (blanc4) to [bend left] (blanc3);
	\end{tikzpicture}
	\caption{The condition of Definition~\ref{def:Good}}
	\label{fig:Good}
\end{figure}
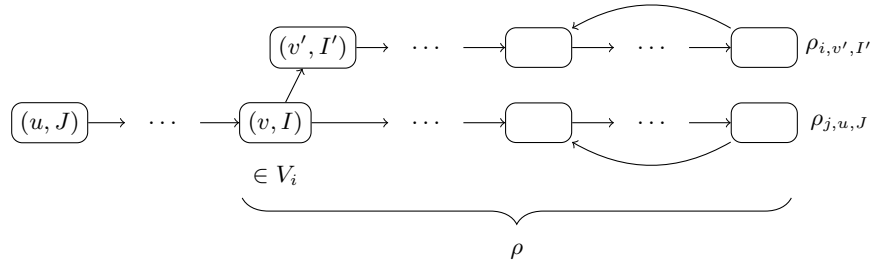

\begin{proposition}
	\label{prop:SPEmemory}
	Let $(\mathcal{G},v_0)$ be a quantitative reachability game and $(\mathcal{X},(v_0,I_0))$ be its extended game, let $c \in (\mathbb{N}\cup\{+\infty\})^{|\Pi|}$ and let $M = \max_{i\in\Pi}\{c_i \mid c_i < +\infty \}$ if this max exists, $M=0$ otherwise. The following assertions are equivalent:
	\begin{enumerate}
		\item There exists an SPE with cost profile $c$ in $(\mathcal{X},(v_0,I_0))$;
		\item $\Lambda^*(v,I) = \{ \rho \in \Plays_{X}(v,I) \mid \rho \text{ is } \lambda^*\text{-consistent}\}$ $\neq \emptyset$ for all $(v,I) \in \Succ^*(v_0,I_0)$ and there exists $\rho \in \Lambda^*(v_0,I_0)$ such that $(\Cost_i(\rho))_{i\in \Pi} = c$;
		\item There exists a good symbolic witness $\mathcal{P}$ that contains a lasso $\rho_{0,v_0,I_0}$ with cost profile $c$ and $|\rho_{0,v_0,I_0}| \leq M + |V|$. Moreover, for each $\rho_{i,v,I} \in \mathcal{P}$, $|\rho_{i,v,I}| \leq \mathcal{O}(|V|^{(|\Pi|+1)\cdot (|\Pi|+|V|)}) + (|\Pi|+1)\cdot |V| $;
		\item There exists a finite-memory SPE $\sigma$ with cost profile $c$ in $(\mathcal{X},(v_0,I_0))$ such that its memory is in $\mathcal{O}(M + 2^{|\Pi|} \cdot |\Pi|\cdot |V|^{(|\Pi|+1)\cdot(|\Pi|+|V|)+1})$.
	\end{enumerate}
\end{proposition}

To obtain the bound on the length of the lassoes, we use the following lemma.

\begin{lemma}[\cite{DBLP:journals/corr/abs-1905.00784}]
	\label{lemma:maxCost}
	Let $v$ be a vertex in the extended game, let $\MaxCost_i(v) = \max \{ \Cost_i(\rho) \mid \rho \in \Plays(v)  \text{ and } \rho \text{ is } \lambda^*\text{-consistent} \}$. If $\MaxCost_i(v)< +\infty$, then  $\MaxCost_i(v) \leq \mathcal{O}(|V|^{(|\Pi|+1)\cdot (|\Pi|+|V|)})$.

\end{lemma}

Notice that the proofs of $(2 \Rightarrow 3)$ and $(3 \Rightarrow 4)$ are quietly the same that the one of Proposition~33 in~\cite{DBLP:journals/corr/abs-1905.00784}. But, in the following proof, we choose more adequatly the plays of the form $\rho_{i,v',I'}$ by picking lassoes with a finite size in the sets of $\lambda^*$-consistent plays beginning in $(v',I')$. It allows us to build and highlight the existent of a symbolic witness which is a symbolic finite representation of an SPE who has cost profile equal to $c$. In this way we build an SPE with a finite memory. 

\begin{proof}[of Proposition~\ref{prop:SPEmemory}]
	
	\noindent$\underline{1 \Rightarrow 2}$: Theorem 17 in~\cite{DBLP:journals/corr/abs-1905.00784}.\\
	$\underline{2 \Rightarrow 3}$: We build a symbolic witness $\mathcal{P}$ step by step and then prove that it is good.
	At the initialization, $\mathcal{P} = \emptyset$. \\
	
	Let $\rho \in \Lambda^*(v_0,I_0)$ such that $(\Cost_i(\rho))_{i\in \Pi} = c$. We apply (P2) on $\rho$ to obtain a lasso $\rho_{0,v_0,I_0}$ such that $|\rho_{0,v_0,I_0}|\leq M +|V|$ and  $(\Cost_i(\rho_{0,v_0,I_0}))_{i\in \Pi} = c$. Moreover, as $\rho$ is $\lambda^*$-consistent, $\rho_{0,v_0,I_0}$ is also $\lambda^*$-consistent (by Lemma~\ref{lem:LassoesConsis}). We add $\rho_{0,v_0,I_0}$ to $\mathcal{P}$.\\
	
	For each $(i,v,I) \in \mathcal{I}$, let $\displaystyle \rho = \argmax_{\rho' \in \Lambda^*(v,I)} \{\Cost_i(\rho')\}$. We obtain $\rho_{i,v,I}$ by copying $\rho$ until Player $i$ has visited his target set, then by removing the unnecessary cycles and apply (P2). If Player $i$ does not visit his target set along $\rho$, we remove all the unnecessary cycles (by apply iteratively (P1)) and then we apply (P2). By the same kind of arguments than for Lemma~\ref{lem:Lassoes} and Lemma~\ref{lem:LassoesConsis}, we obtain that: \emph{i)} $\rho_{i,v,I}$ is $\lambda^*$-consistent, \emph{ii)} $\Cost_i(\rho_{i,v,I}) = \Cost_i(\rho)$ and \emph{iii)} $|\rho_{i,v,I}| \leq \mathcal{O}(|V|^{(|\Pi|+1)\cdot (|\Pi|+|V|)}) + (|\Pi| + 1) \cdot |V|$ (by Lemma~\ref{lemma:maxCost}). We add $\rho_{i,v,I}$ to $\mathcal{P}$. \\
	
	By construction $\mathcal{P}$ is a symbolic witness. It remains to prove that it is good.
	
	Let $\rho_{j,u,J}$ and $\rho_{i,v',I'} \in \mathcal{P}$  and let $\rho \in \Plays(v,I)$ be a suffix of $\rho_{j,u,J}$ such that $((v,I),(v',I')) \in E^X$ and $(v,I) \in V^X_i$, we have to prove that if $i \not \in I$, $\Cost_i(\rho) \leq 1 + \Cost_i(\rho_{i,v',I'})$.
	
	As $i \not \in I$, $ \displaystyle \lambda^{*+1}(v,I)= 1+ \min_{((v,I),(v',I')\in E^X)} \{ \Cost_i(\rho') \mid \rho' \in \Lambda^*(v',I') \} $, we have that $\lambda^*(v,I) = \lambda^{*+1}(v,I) \leq 1 + \Cost_i(\rho_{i,v',I'})$. Finally, as $\rho_{j,u,J}$ is $\lambda^*$-consistent, we have: 
	$$ \Cost_i(\rho) \leq \lambda^*(v,I) \leq 1 + \Cost_i(\rho_{i,v',I'}).$$
	It concludes the proof.\\
	$\underline{3 \Rightarrow 4}$: Given a good symbolic witness $\mathcal{P}$ with properties given in statement 3, we show how to build an SPE with a finite memory. 
	
	We define a strategy profile $\sigma$ step by step by induction on the subgames of $(\mathcal{X}, (v_0,I_0))$. We first partially define $\sigma$ such that $\outcome{\sigma}{(v_0,I_0)} = \rho_{0,v_0,I_0}$.
	
	Consider next, $h(v,I)(v',I') \in \Hist(v_0,I_0)$ with $(v,I) \in V_i^X$ such that $\outcome{\sigma_{\restriction h}}{(v,I)}$ is already built but not $\outcome{\sigma_{\restriction h(v,I)}}{(v',I')}$. Then we extend $\sigma$ such that $\outcome{\sigma_{\restriction h (v,I)}}{(v',I')} = \rho_{i,v',I'}$. 
	
	Let us prove that $\sigma$ is a very weak SPE (and so an SPE). Consider the subgame $(\mathcal{X}_{\restriction h},(v,I))$ (with $(v,I) \in V_i^X$) and the one-shot  deviating strategy $\sigma'_i$ from $\sigma_{i\restriction h}$ such that $\sigma'_i(v,I)= (v',I')$.
	By construction, there exists $\rho_{j,u,J}$ and $\rho_{i,v',I'}\in \mathcal{P}$ and $h' \in \Hist(v_0,I_0)$ such that $h\outcome{\sigma_{\restriction h}}{(v,I)}= h'\rho_{j,u,J}$ and $\outcome{\sigma_{\restriction h(v,I)}}{(v',I')} = \rho_{i,v',I'}.$
	 We have to prove that:
	
	$$ \Cost_i(h \outcome{\sigma_{\restriction h}}{(v,I)}) \leq \Cost_i(h(v,I)\outcome{\sigma_{\restriction h(v,I)}}{(v',I')}).$$
	
	If $i \in I$, then $\Cost_i(h \outcome{\sigma_{\restriction h}}{(v,I)}) = \Cost_i(h(v,I)\outcome{\sigma_{\restriction h(v,I)}}{(v',I')})=0$. Otherwise, 
	
	\begin{align*}
		\Cost_i(h\outcome{\sigma_{\restriction h}}{(v,I)}) &= \Cost_i(h'\rho_{j,u,J}) \\
														   &=|h(v,I)|+\Cost_i(\rho)& (\text{where }\rho \text{ is a suffix of }\rho_{j,u,J} \text{ beginning in }(v,I))\\
														&\leq |h(v,I)|+1+\Cost_i(\rho_{i,v',I'}) & (\mathcal{P} \text{ is a good symbolic witness})\\
														&= \Cost_i(h(v,I)\rho_{i,v',I'})\\
														&= \Cost_i(h(v,I)\outcome{\sigma_{\restriction h(v,I)}}{v',I'}).
\end{align*}
Notice that $\sigma$ has a cost profile equal to $c$ by construction, as $(\Cost_i(\rho_{0,v_0,I_0}))_{i\in\Pi} = c$. It remains to prove that $\sigma$ is finite-memory with memory in $\mathcal{O}(M + 2^{|\Pi|} \cdot |\Pi|\cdot |V|^{(|\Pi|+1)\cdot(|\Pi|+|V|)+1})$. Having $(j,u,J)$ in memory (the last deviating player $j$ and the vertex $(u,J)$ where he moved), the machine $\mathcal{M}_i$, $i\in \Pi$, which represents the strategy $\sigma_i$, has to produce the lasso $\rho_{j,u,J}$ of length bounded by $M + (|\Pi|+1)\cdot |V|$  for $\rho_{0,v_0,I_0}$ and  by $\mathcal{O}(|V|^{(|\Pi|+1)\cdot (|\Pi|+|V|)})$ for the others (at most $|\Pi|\cdot|V| \cdot 2^{|\Pi|}$ such lassoes). It leads to a memory in $\mathcal{O}(M + 2^{|\Pi|} \cdot |\Pi|\cdot |V|^{(|\Pi|+1)\cdot(|\Pi|+|V|)+1})$. \\
	$\underline{4 \Rightarrow 1}$: Obvious. \qed
	
\end{proof}

\begin{proof}[of Corollary~\ref{cor:corCritOut2ConstraintProb} for SPEs]
	
By hypothesis, we know that there exists an SPE $\sigma$ such that $(\Cost_i(\outcome{\sigma}{x_0})) \leq y$. 
Let $\rho = \outcome{\sigma}{x_0}$, and apply (P1) as long as possible and then (P2) in order to obtain a lasso $\rho'=h\ell^\omega$ with size at most $(|\Pi|+1)\cdot |V|$  and such that $\Cost_i(\rho')\leq \min \{|V|\cdot|\Pi|, y_i\}$ (Lemma~\ref{lem:Lassoes}). As $\rho$ is the outcome of an SPE, $\rho$ is $\lambda^*$-consistent and thus $\rho'$ is also $\lambda^*$-consistent (Lemma~\ref{lem:LassoesConsis}). Thus $\rho'$ is the outcome of an SPE.

Thanks to Proposition~\ref{prop:SPEmemory}, we know that there exists an SPE $\tau$ with finite-memory in $\mathcal{O}(2^{|\Pi|} \cdot |\Pi|\cdot |V|^{(|\Pi|+1)\cdot(|\Pi|+|V|)+1})$ which as the same cost profile as $\rho'$. So, $\Cost_i(\outcome{\tau}{x_0}) \leq |V|\cdot|\Pi|$ for each $i \in \Visit(\outcome{\tau}{x_0})$ and $\tau$ fulfills the constraints. Moreover,  one can  assume that $\outcome{\tau}{x_0}$ is a lasso of length at most $(|\Pi|+1)\cdot |V|$. \qed
\end{proof}

\subsection{Proof of Proposition~\ref{prop:particularPareto}}

\begin{proof}[of Proposition~\ref{prop:particularPareto}]
	Second item is a direct consequence of the first one. Thus, let us prove first item.
	
	Let $\sigma$ be an NE such that its cost profile is Pareto optimal in $\Plays(v_0)$. To get a contradiction, assume that there exists $i \in \Visit(\outcome{\sigma}{v_0})$ such that $ \Cost_i(\outcome{\sigma}{v_0}) > |V|\cdot |\Pi|$. It means that there exists an unnecessary cycle before Player $i$ reaches his target set. By removing this cycle (applying (P1)), we obtain a new play $\rho'$ such that $\Cost_i(\rho') < \Cost_i(\outcome{\sigma}{v_0})$ and for Player $j$ ($j \neq i$), $\Cost_j(\rho') \leq \Cost_j(\outcome{\sigma}{v_0})$ (by Lemma~\ref{lem:Lassoes}). It leads to a contradiction with the fact that $(\Cost_i(\outcome{\sigma}{v_0}))_{i \in \Pi}$ is Pareto optimal in $\Plays(v_0)$.
	
	The same proof holds for SPE. \qed
\end{proof}

\section{Complements of Section~\ref{subsection:results}}


\subsection{Proof of Theorem~\ref{thm:Complexity} and Theorem~\ref{thm:Memory} for Problem~\ref{prob:decisionThreshold}}

\begin{proposition}
	\label{prop:compConstNashEasy}
	Let $(\mathcal{G},v_0)$ be a quantitative reachability game, for NE:
	\begin{itemize}
		\item Problem~\ref{prob:decisionThreshold} is NP-easy.
		\item If the answer to this decision problem is positive, then there exists an NE $\sigma$ with memory in $\mathcal{O}((|\Pi|+1)\cdot |V|)$ which satisfies the constraints and such that for all $i \in \Visit(\outcome{\sigma}{v_0})$, $\Cost_i(\outcome{\sigma}{v_0}) \leq |V|\cdot |\Pi|$.
	\end{itemize}
\end{proposition}
\begin{proof}
	Let $(\mathcal{G},v_0)$ be a quantitative reachability game and let $y \in (\mathbb{N}\cup\{+\infty\})^{|\Pi|}$ be a threshold. By Corollary~\ref{cor:corCritOut2ConstraintProb}, we know that if there exists an NE $\sigma$ which fulfills the constraints, then there exists an other one such that its outcome is of the form $h\ell^\omega$ with $|h\ell|\leq (|\Pi|+1)\cdot |V| $ and which fulfills the constraints too.\\
	
	 Thus, the NP-algorithm is the following one:\\
		\textbf{Step 1:} Guess a lasso $\rho = h\ell^\omega$ with size at most $(|\Pi|+1)\cdot |V|$ such that $\Visit(\rho)= \Visit(h)$ and verify that $\rho$ is a lasso.\\
		\textbf{Step 2:} Compute $(\Cost_i(\rho))_{i\in \Pi}$ and verify that $(\Cost_i(\rho))_{i\in \Pi} \leq y$.\\
		\textbf{Step 3:} For each $i \in \Pi$, compute the values in $\mathcal{G}_i$.\\
		\textbf{Step 4:} Verify that $\rho$ is the outcome of an NE (thanks to Theorem~\ref{thm:critOutcomeNE}): as no new player reaches his target set along $\ell$ and $\ell$ is repeated infinitely often at the end of the outcome, we only have to check that $\rho$ is $\ConstNE$-consistent along $h\ell$ (an not along $h\ell^\omega$).\\
		
		The second item is a direct consequence of Corollary~\ref{cor:corCritOut2ConstraintProb}.\qed
\end{proof}

To obtain the NP-completeness it remains to prove the NP-hardness.

\begin{proposition}
	\label{prop:compConstNashHard}	
	Let $(\mathcal{G},v_0)$ be a quantitative reachability game, for NE, Problem~\ref{prob:decisionThreshold} is NP-hard.
\end{proposition}
This proof is based on a polynomial reduction from SAT and is inspired by the reduction provides for safety objective in~\cite{condurache_et_al:LIPIcs:2016:6256}.
\begin{proof}
	To prove this proposition, we give a polynomial reduction from the SAT problem that is NP-complete. Let $X = \{x_1, \ldots, x_m \}$ be the set of variables and $\psi = C_1\wedge\ldots \wedge C_n$ is a Boolean formula in CNF over $X$ and equals to the conjunction of the clauses $C_1, \ldots, C_n$. This problem is to decide if the formula $\psi$ is true. Such a formula is true if there exists a valuation $\mathcal{I} : X \rightarrow \{0,1\}$ such that the valuation of $X$ with respect to $\mathcal{I}$ evaluates $\psi$ to true.
	
	We build the following quantitative reachability game $\mathcal{G}_{\psi}= (\mathcal{A}, (\Cost_i)_{i \in \Pi}, (F_i)_{i\in \Pi})$ where $\mathcal{A} = (\Pi, V, E, (V_i)_{i \in \Pi})$:
	\begin{itemize}
		\item the arena $\arena$ is depicted in Figure~\ref{fig:redFormToGame} where the set of players $\Pi = \{1, \ldots, n, n+1 \}$ has $n+1$ players: one by clause (players $1$ to $n$) and an additional one. $V_{n+1}$ is depicted by squared vertices and for all $1 \leq i \leq n$, $V_i = \{P_i \}$;
		\item $F_{n+1} = \{ T_w \}$;
		\item for all $1 \leq i \leq n$, $F_i = \{0_x \mid \neg x \in C_i \} \cup \{ 1_x \mid x \in C_i \}\cup \{ T_\ell \}.$
	\end{itemize}
	
	The game $\mathcal{G}_{\psi}$ can be build from $\psi$ in polynomial time. Let us show that $\psi$ is true if and only if there exists an NE $ \sigma $ in $(\mathcal{G}_\psi,x_1)$ such that $(\Cost_i(\outcome{\sigma}{x_1}))_{i\in\Pi} \leq (2m,\ldots,2m,2m +n)$.\\
	
	$(\Rightarrow)$ Suppose that $\psi$ is true. Then there exists $\mathcal{I}: X \rightarrow \{0,1\}$ such that the valuation of $\psi$ with respect to $\mathcal{I}$ evaluates $\psi$ to true.
	
	Let us consider $\sigma$ defined in this way:
	
	\begin{itemize}
		\item for all $hv \in \Hist_{n+1}(x_1)$, $\sigma_{n+1}(hv) = \begin{cases} \mathcal{I}(v) & \text{ if } v = x \text{ with } x \in X \\
		 																		x_{k+1} & \text{ if } v = 0_{x_k} \text{ or } 1_{x_k} \text{with } 0 \leq k \leq m-1 \\
		P_1 & \text{if } v = 0_{x_m} \text{ or } 1_{x_m} \\
		T_\ell & \text{if } v = T_\ell \\
		T_w & \text{if } v = T_w\end{cases}$.
		\item for all $1 \leq i \leq n-1 $ (resp. for $i = n$), for all $hv \in \Hist_i(x_1)$ $\sigma_i(hv) = P_{i+1} $ (resp. $T_w$) if $hv$ is consistent with $\sigma_{n+1}$ and $\sigma_i(hv)= T_\ell$ otherwise. 	
		
	\end{itemize}
	
	We now prove that $\sigma$ is an NE. It is clear that $\outcome{\sigma}{x_1}$ is of the form $hP_1P_2\ldots P_n T_w^\omega$ and as $\sigma_{n+1}$ corresponds to the valuation $\mathcal{I}$ which evaluates $\psi$ to true, we have that : $\Cost_{n+1}(\outcome{\sigma}{x_1})= 2m+n$ and for all $1\leq i \leq n$, $\Cost_i(\outcome{\sigma}{x_1})\leq 2m$. Obviously Player $n+1$ does not have an incentive to deviate from $\sigma_{n+1}$ because it is a the least cost that he can obtain. For each player $i$ such that $1 \leq i \leq n$, as Player $i$ reaches his target set before he can play, changing his strategy does not change his cost. Thus, no player has a profitable deviation and $\sigma$ is an NE.\\
	
	$(\Leftarrow)$ Suppose that there exists an NE $\sigma$ in $(\mathcal{G}_{\psi}, x_1)$ such that $\Cost_i(\outcome{\sigma}{x_1}) \leq (2m, \ldots, 2m, 2m+n)$ and let us prove that $\psi$ is true.
	
	We define $\mathcal{I} : X \rightarrow \{0,1\}$ as follows: for $x \in X$, $\mathcal{I}(x) = \sigma_{n+1}(hx)$ with $hx$ a prefix of $\outcome{\sigma}{x_1}$. Let us show that the valuation $\mathcal{I}$ evaluates $\psi$ to true. As $\Cost_{n+1}(\outcome{\sigma}{x_1}) < +\infty$, it means that Player $n+1$ visits his target set $\{T_{w}\}$. Additionaly, for all $1\leq i \leq n$, as $\Cost_i(\outcome{\sigma}{x_1}) < +\infty$, Player $i$ visits also his target set but not $T_\ell$. Thus, he reaches $\{0_x \mid \neg x \in C_i \} \cup \{ 1_x \mid x \in C_i \}$ and this means that the clause $C_i$ is true if each variable in $C_i$ is replaced by its valuation. As it is the case for each clause, the formula $\psi$ is true.

	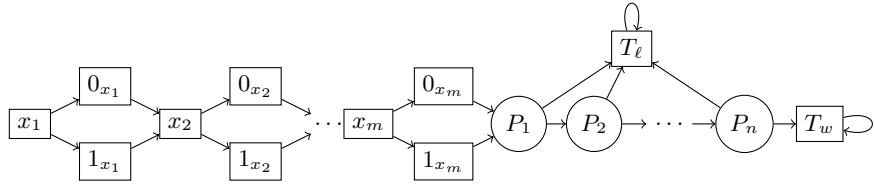
\begin{figure}

		\centering
		\begin{tikzpicture}
			\node[draw] (x1) at (0,0){$x_1$};
			\node[draw] (x10) at (1, 0.5){$0_{x_1}$};
			\node[draw] (x11) at (1, -0.5){$1_{x_1}$};
			
			\draw[->] (x1) to (x10);
			\draw[->] (x1) to (x11);
			
			\node[draw] (x2) at (2,0){$x_2$};
			\node[draw] (x20) at (3,0.5){$0_{x_2}$};
			\node[draw] (x21) at (3,-0.5){$1_{x_2}$};
			
			\draw[->] (x10) to (x2);
			\draw[->] (x11) to (x2);
			\draw[->] (x2) to (x20);
			\draw[->] (x2) to (x21);
			
			\node (int1) at (4,0){\ldots};
			\draw[->] (x20) to (int1);
			\draw[->] (x21) to (int1);
			
			\node[draw] (xm) at (4.5,0){$x_m$};
			\node[draw] (xm0) at (5.5,0.5){$0_{x_m}$};
			\node[draw] (xm1) at (5.5,-0.5){$1_{x_m}$};
			
			\draw[->] (xm) to (xm0);
			\draw[->] (xm) to (xm1);
			
			\node[draw,circle] (P1) at (6.5,0){$P_1$};
			
			\draw[->] (xm1) to (P1);
			\draw[->] (xm0) to (P1);
			
			\node[draw,circle] (P2) at (7.5,0){$P_2$};
			
			\draw[->] (P1) to (P2);
			
			\node (int2) at (8.5,0){$\ldots$};
			\draw[->] (P2) to (int2);
			
			\node[draw,circle] (Pn) at (9.5,0){$P_n$};
			\draw[->] (int2) to (Pn);
			
			\node[draw] (Tw) at (10.5,0){$T_{w}$};
			\draw[->] (Pn) to (Tw);
			
			\draw[->] (Tw) to [loop right] (Tw);
			
			\node[draw] (Tl) at (8,1){$T_{\ell}$};

			\draw[->] (Tl) to [loop above] (Tl);
			
			\draw[->] (P1) to (Tl);
			\draw[->] (P2) to (Tl);
			\draw[->] (Pn) to (Tl);			
		\end{tikzpicture}
			\caption{Reduction from the formula $\psi$ to the quantitative reachability game $\mathcal{G}_{\psi}$}
			\label{fig:redFormToGame}
	\end{figure}
\end{proof}

We conclude results about the threshold existence problem in quantitative reachability games with some remarks about a variant of this problem in this setting.

\begin{remark}
	\label{rem:twoThresholds}
	We may also consider this problem but with an upper and a lower threshold. This problem is also NP-complete but the result about memory is a little bit different. Indeed, the memory depends on the upper bound. If there exists an NE which satisfies the constraints then there exists an other one which satisfies the constraints too but with memory in $\mathcal{O}(\max_{i \in \Pi \mid y_i \neq +\infty} y_i + |V| + |\Pi|)$. The difference is due to the fact that we cannot apply iteratively the procedure (P1) and then procedure (P2) to obtain a lasso as procedure (P1) can decrease the cost of the outcome. So the new outcome could  no longer satisfy the constraints. To solve this problem, we only have to apply procedure (P2) on the outcome of $\sigma$ and obtain a lasso of length at most $\max_{i \in \Pi \mid y_i \neq +\infty} y_i + |V|$.
\end{remark}

\begin{theorem}
	\label{thm:thrSPE}
	Let $(\mathcal{G},v_0)$ be a quantitative reachability game and $(\mathcal{X},x_0)$ be its extended game,
	\begin{itemize}
		\item For SPE, (\cite{DBLP:journals/corr/abs-1905.00784}) Problem~\ref{prob:decisionThreshold} is PSPACE-complete.
		\item  If the answer of the decision problem is positive, there exists a strategy profile $ \sigma$ with memory in  $\mathcal{O}(2^{|\Pi|} \cdot |\Pi|\cdot |V|^{(|\Pi|+1)\cdot(|\Pi|+|V|)+1})$ which satisfies the constraints and such that $\Cost_i(\outcome{\sigma}{x_0}) \leq |V| \cdot |\Pi|$ if $i\in \Visit(\outcome{\sigma}{x_0})$.
	\end{itemize}
\end{theorem}

Notice that, in~\cite{DBLP:journals/corr/abs-1905.00784}, Problem~\ref{prob:decisionThreshold} for SPE in quantitative reachability games is shown PSPACE-complete with a lower and an upper threshold.

\begin{proof}
	Let us prove the second assertion. It is a direct consequence of Corollary~\ref{cor:corCritOut2ConstraintProb}. \qed
\end{proof}


\subsection{Proof of Theorem~\ref{thm:Complexity} and Theorem~\ref{thm:Memory} for Problem~\ref{prob:decisionWelfare}}

\begin{lemma}
	\label{lem:NEwelfare} Let $(\mathcal{G},v_0)$  be a reachability game (resp.$(\mathcal{X},x_0)$ its extended game) and let $w_0= v_0$ (resp. $w_0=x_0$)
	Let $k \in \{0, \ldots, |\Pi| \}$ and $c \in \mathbb{N}$ be two thresholds, let $p \in (\mathbb{N}\cup \{+\infty\})^{|\Pi|}$ be a cost profile and let $\sigma$ be a strategy profile in $(\mathcal{G},v_0)$ (resp. $(\mathcal{X},x_0)$). We define $\mathcal{R} = \{ i \in \Pi \mid p_i < +\infty \}$. If $(|\mathcal{R}|, \sum_{i\in \mathcal{R}} p_i) \succeq (k,c)$ and $(\Cost_i(\outcome{\sigma}{w_0}))_{i\in\Pi} \leq p$ then $\SW(\outcome{\sigma}{w_0}) \succeq (k,c)$.
\end{lemma}

\begin{proof}
	
	Assume that $(|\mathcal{R}|, \sum_{i\in \mathcal{R}} p_i) \succeq (k,c)$ and $(\Cost_i(\outcome{\sigma}{w_0}))_{i\in \Pi} \leq p$. As $(\Cost_i(\outcome{\sigma}{w_0}))_{i\in \Pi} \leq p$, we have that $|\Visit(\outcome{\sigma}{w_0})| \geq |\mathcal{R}|$.
	
	 Thus, $\SW(\outcome{\sigma}{w_0}) \succeq (|\mathcal{R}|, \sum_{i\in \mathcal{R}} p_i)$ and as the lexicographic ordering is transitive we have: $\SW(\outcome{\sigma}{v_0}) \succeq (k,c)$. \qed
	
\end{proof}

\begin{proposition}
	\label{prop:quantNEWelfare}
	Let $(\mathcal{G},v_0)$ be a quantitative reachability game. 
	\begin{itemize}
		\item For NE, Problem~\ref{prob:decisionWelfare} is NP-complete.
		\item If the answer to this decision problem is affirmative, then there exists a strategy profile $\sigma$ with memory in  $\mathcal{O}((|\Pi|+1)\cdot |V|)$ which satisfies the constraints and such $\sum_{i \in \Visit(\outcome{\sigma}{v_0})} \Cost_i(\outcome{\sigma}{v_0}) \leq |\Pi|^2 \cdot |V|$ . 
	\end{itemize}
\end{proposition}

\begin{proof}
	
	Let $(\mathcal{G},v_0)$ be a quantitative reachability game and let $k \in \{0, \ldots, |\Pi| \}$ and $c \in \mathbb{N}$ be two thresholds. 
	
	The NP-algorithm works as follows:\\	
	\textbf{Step 1:} Guess $p \in (\mathbb{N}\cup \{+\infty\})^{|\Pi|}$ a cost profile;\\
 	\textbf{Step 2:} Let $\mathcal{R} = \{ i \in \Pi \mid p_i < +\infty \}$ be the set of players who have a finite cost in the cost profile $p$, verify if $(|\mathcal{R}|, \sum_{i \in \mathcal{R}}p_i) \succeq (k,c)$;\\
	\textbf{Step 3:} Verify that there exists an NE $\sigma$ in $(\mathcal{G},v_0)$ such that $(\Cost_i(\outcome{\sigma}{v_0}))_{i\in\Pi} \leq p. (\star)$\\
	
	Thanks to Proposition~\ref{prop:compConstNashEasy} and Lemma~\ref{lem:NEwelfare}, this algorithm is an NP-algorithm to solve Problem~\ref{prob:decisionWelfare} for Nash equilibrium. Moreover, if there exists an NE which fulfills constraints $(\star)$ there exists one with a polynomial memory and with costs less or equal to $|\Pi|\cdot|V|$ for players who have visited their target set (see Corollary~\ref{cor:corCritOut2ConstraintProb}). This leads to an NE such that the accumulated cost of the players who have visited their target set is less or equal to $|\Pi|^2\cdot |V|$ and it also fulfills the constraints for the social welfare.\\
	
The NP-hardness is due to a polynomial reduction from the  SAT problem in the same philosophy than the one for Problem~\ref{prob:decisionThreshold} for NE in quantitative reachability games: the SAT-formula is satisfiable if and only if the exists an NE with a social welfare $\succeq (|\Pi|,S)$ where $S = 2 m n + 2m+n$ (sum of the components of the threshold fixed for the reduction for Problem~\ref{prob:decisionThreshold}). \qed
\end{proof}

\begin{theorem}
	\label{thm:quantSPEWelfare}
	Let $(\mathcal{G},v_0)$ be a quantitative reachability game and $(\mathcal{X},x_0)$ be its extended game.
	\begin{itemize}
		\item For SPE, Problem~\ref{prob:decisionWelfare} is PSPACE-complete.
		\item If the answer to this decision problem is positive, then there exists a strategy profile $\sigma$ with memory in  $\mathcal{O}(2^{|\Pi|} \cdot |\Pi|\cdot |V|^{(|\Pi|+1)\cdot(|\Pi|+|V|)+1})$ which satisfies the constraints and such $\sum_{i \in \Visit(\outcome{\sigma}{x_0})} \Cost_i(\outcome{\sigma}{x_0}) \leq |\Pi|^2 \cdot |V|$ .
	\end{itemize}
\end{theorem}
\begin{proof}
	Let $(\mathcal{G},v_0)$ be a quantitative reachability game and $(\mathcal{X},x_0)$ be its extended game and let $k \in \{0, \ldots, |\Pi| \}$ and $c \in \mathbb{N}$ be two thresholds. 
	
	The PSPACE-algorithm works as follows:\\	
	\textbf{Step 1:} Guess $p \in (\mathbb{N}\cup \{+\infty\})^{|\Pi|}$ a cost profile;\\
 	\textbf{Step 2:}Let $\mathcal{R} = \{ i \in \Pi \mid p_i < +\infty \}$ be the set of players who have a finite cost in the cost profile $p$, verify if $(|\mathcal{R}|, \sum_{i \in \mathcal{R}}p_i) \succeq (k,c)$;\\
	\textbf{Step 3:} Verify that there exists an SPE $\sigma$ in $(\mathcal{X}, x_0)$ such that $(\Cost_i(\outcome{\sigma}{x_0}))_{i\in\Pi} \leq p.$ ($\star$)\\
	
	Thanks to Theorem~\ref{thm:thrSPE} and Lemma~\ref{lem:NEwelfare}, this algorithm is a PSPACE-algorithm to solve Problem~\ref{prob:decisionWelfare} for SPE. Moreover, if there exists an SPE which fulfills the constraints $(\star)$ there exists one with a finite-memory and with costs less or equal to $|\Pi|\cdot|V|$ for players who have visited their target set (see Corollary~\ref{cor:corCritOut2ConstraintProb}). This leads to an SPE such that the cost sum of the players who have visited their target set is less or equal to $|\Pi|^2\cdot |V|$ and it fulfills the constraints $(\star)$ and thus the constraints of Problem~\ref{prob:decisionWelfare} (by Lemma~\ref{lem:NEwelfare}).\\
	
The PSPACE-hardness is due to a polynomial reduction from the QBF problem (which is PSPACE-complete) in the same philosophy than the one for Problem~\ref{prob:decisionThreshold} for SPE in quantitative reachability games (see~\cite{DBLP:journals/corr/abs-1905.00784}): the fully quantified Boolean formula is satisfiable if and only if the exists an SPE with a social welfare $\succeq (|\Pi|-1,S)$ where $S = 2 m n + 2m+n$ (sum of the components of the threshold fixed for the reduction for Problem~\ref{prob:decisionThreshold} for SPE in quantitative reachability games). But, if we want that the implication ``$\Leftarrow$'' holds, we have to slightly change the arena of the game (see Figure~\ref{figure:reachPSPACEh}): we add a vertex $\perp$ that we add in the target set of the ``existential player'' (Player $n+1$) and that is reachable from  $q_1 \in V_{n+1}$ (we can assume that the formula $\psi = Q_1x_1Q_2x_2\ldots Q_mx_m \phi(X)$ is such that $Q_k = \exists$ if $k$ is odd and $Q_k = \forall$ otherwise). Notice that the weight $2S$ on the edge $(q_1,\perp)$ should be understood as a path of length $2S$.

The main idea of the implication ``$\Leftarrow$'' is the following one: we have an SPE $\sigma$ sucht that $\SW(\outcome{\sigma}{q_1}) \succeq (|\Pi|-1, S)$. As Player $n+1$ visits his target set if and only if Player $n+2$ does not, we have that $|\Visit(\outcome{\sigma}{q_1})| = |\Pi|-1$. Assume that Player $n+1$ does not visit his target set, as the game is initialized in $q_1$ which is a vertex of Player $n+1$, he can go to $\perp$ and in this way he visits his target set. So, it is a profitable deviation for Player $n+1$ and it leads to a contradiction with the fact that $\sigma$ is an SPE. We can conclude that $\sigma$ is an SPE such that all players excepted Player $n+2$ visit their target set.

Additionaly, for each $\rho \in \Plays(q_1)$, if $\Cost_{n+1}(\rho) < +\infty$ (either the vertex $\perp$ or the vertex $t_{n+1}$ is reached) then no vertex $t_k$ with $1\leq k \leq n$ is reached.  Then the philosophy of the proof is the same as the one in~\cite{DBLP:journals/corr/abs-1905.00784}.

Notice that the weight on the edge $(q_1,\perp)$ allows to ensure that, for implication ``$\Rightarrow$'', Player $n+1$ does not have an incentive to go to $\perp$ in place to follow the valuation that makes the formula true.

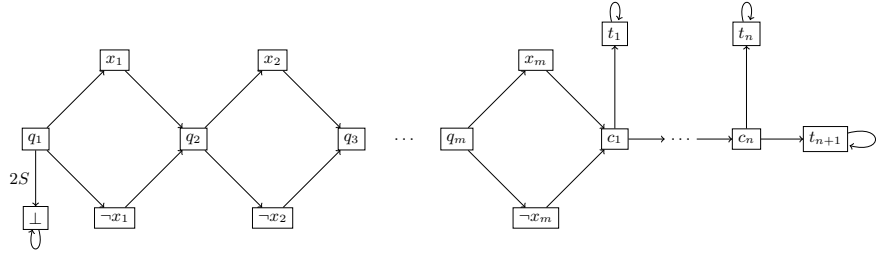
\begin{figure}[h!]
\centering
\scalebox{0.7}{
\begin{tikzpicture}
\node[draw] (Q1) at (0,0){$q_1$};
\node[draw] (perp) at (0,-1.5){$\perp$};
\node[draw] (Q2) at (3,0){$q_2$};
\node[draw] (Q3) at (6,0){$q_3$};
\node (empty) at (7,0){$\ldots$};
\node[draw] (Qm) at (8,0){$q_m$};
\node[draw] (C1) at (11,0){$c_1$};
\node (empty2) at (12.25,0){$\ldots$};
\node[draw](Cn) at (13.5,0){$c_n$};
\node[draw](T0) at (15,0){$t_{n+1}$};

\node[draw] (x1) at (1.5,1.5){$x_1$};
\node[draw] (nx1) at (1.5,-1.5){$\neg x_1$};

\node[draw] (x2) at (4.5,1.5){$x_2$};
\node[draw] (nx2) at (4.5,-1.5){$\neg x_2$};

\node[draw] (xm) at (9.5, 1.5){$x_m$};
\node[draw] (nxm) at (9.5, -1.5){$\neg x_m$};

\node[draw] (T1) at (11,2){$t_1$};
\node[draw] (Tn) at (13.5,2){$t_n$};

\draw[->] (Q1) to (x1);
\draw[->] (Q1) to (nx1);
\draw[->] (Q1) to node[left]{$2S$} (perp);
\draw[->] (perp) to [loop below] (perp);

\draw[->] (x1) to (Q2);
\draw[->] (nx1) to (Q2);

\draw[->] (Q2) to (x2);
\draw[->] (Q2) to (nx2);

\draw[->] (x2) to (Q3);
\draw[->] (nx2) to (Q3);

\draw[->](Qm) to (xm);
\draw[->](Qm) to (nxm);

\draw[->](xm) to (C1);
\draw[->](nxm) to (C1);

\draw[->](C1) to (empty2);
\draw[->](empty2) to (Cn);

\draw[->](Cn) to (T0);

\draw[->](C1) to (T1);
\draw[->](Cn) to (Tn);

\draw[->] (T1) edge [loop above] (T1);
\draw[->] (Tn) edge [loop above] (Tn);
\draw[->] (T0) edge [loop right] (T0);
\end{tikzpicture}}
\caption{Reduction from the formula $\psi$ to the quantitative reachability game $\mathcal{G}_{\psi}$ for Problem~\ref{prob:decisionWelfare} with SPEs.}
\label{figure:reachPSPACEh}
\end{figure}  \qed
\end{proof}


\subsection{Proof of Theorem~\ref{thm:Complexity} and Theorem~\ref{thm:Memory} for Problem~\ref{prob:decisionPareto}}

\begin{lemma}
	Problem~\ref{prob:4} belongs to co-NP for quantitative reachability games.
\end{lemma}
\begin{proof}
	Let us prove it for quantitative reachability games. If $\rho$ is not Pareto optimal, there exists a play $\rho'$ such that $(\Cost_i(\rho))_{i\in\Pi} \geq (\Cost_i(\rho'))_{i\in\Pi}$ and $(\Cost_i(\rho))_{i\in\Pi} \neq (\Cost_i(\rho'))_{i\in\Pi}$. Moreover, thanks to Lemma~\ref{lem:Lassoes}, one may assume that $\rho'$ is a lasso with size at most $(|\Pi|+1)\cdot |V|$. So, we only have to guess such a lasso $\rho'$ and to verify that $(\Cost_i(\rho))_{i\in\Pi} \geq (\Cost_i(\rho'))_{i\in\Pi}$ and $(\Cost_i(\rho))_{i\in\Pi} \neq (\Cost_i(\rho'))_{i\in\Pi}$. This can be done in polynomial time.
	
\end{proof}

\begin{proposition}
		Let $(\mathcal{G},v_0)$ be a quantitative reachability game, for NE:
		
		\begin{itemize}
			\item Problem~\ref{prob:decisionPareto} belongs to $\Sigma^P_2$ and is NP-hard.
			\item If the answer to this decision problem is positive, then there exists an NE $\sigma$ with memory in $\mathcal{O}((|\Pi|+1)\cdot |V|)$ such that for all $i \in \Visit(\outcome{\sigma}{v_0})$, $\Cost_i(\outcome{\sigma}{v_0}) \leq |\Pi|\cdot|V|$.
		\end{itemize}
\end{proposition}
\begin{proof}
	We can provide the following $\Sigma^P_2$-algorithm for Problem~\ref{prob:decisionPareto}, given an oracle for Problem~\ref{prob:4}:\\
	 \textbf{Step 1:} Guess a play $\rho$ as a lasso of length at most $(|\Pi|+1)\cdot|V|$ (sufficient thanks to Proposition~\ref{prop:particularPareto}).\\
	 \textbf{Step 2:} Check that $(\Cost_i(\rho))_{i\in \Pi}$ is Pareto optimal in $\Plays(v_0)$ using the oracle for Problem~\ref{prob:4}.\\
	 \textbf{Step 3:} Check that $\rho$ is the outcome of a Nash equilibrium using the characterization.\\
	
	Additionally, as $\rho$ is a lasso with length at most $(|\Pi|+1)\cdot|V|$, it provides an NE with outcome $\rho$ and memory in $\mathcal{O}((|\Pi|+1)\cdot |V|)$ such that for all $i \in \Visit(\rho)$, $\Cost_i(\rho) \leq |V|\cdot |\Pi|$ (by Proposition~\ref{prop:particularPareto}).\\
	
	The NP-hardness is due to a polynomial reduction from the SAT problem which is NP-complete. The philosophy of the proof is the same than the one for the qualitative setting (see proof of Proposition~\ref{prop:qualNEPareto}) but we have to put a sufficiently high weight on the edge between the initial vertex and the vertex $\perp$ to ensure that Player $n+1$ does not have an incentive to go to $\perp$.\qed
\end{proof}

\begin{proposition}
		Let $(\mathcal{G},v_0)$ be a quantitative reachability game and let $(\mathcal{X},x_0)$ its extended game, for SPE:
		
		\begin{itemize}
			\item Problem~\ref{prob:decisionPareto} is PSPACE-complete.
			\item If the answer to this decision problem is positive, then there exists an SPE $\sigma$ with memory with memory in  $\mathcal{O}(2^{|\Pi|} \cdot |\Pi|\cdot |V|^{(|\Pi|+1)\cdot(|\Pi|+|V|)+1})$ which satisfies the constraints  such that for all $i \in \Visit(\outcome{\sigma}{x_0})$, $\Cost_i(\outcome{\sigma}{x_0}) \leq |\Pi|\cdot|V|$.
		\end{itemize}
\end{proposition}
\begin{proof}
	The PSPACE-algorithm works as follows\\
	 \textbf{Step 1:} Guess a cost profile $p \in (\mathbb{N}\cup \{+\infty\})^{|\Pi|}$;\\
	 \textbf{Step 2:} Check that $p$ is Pareto optimal in $\Plays(x_0)$ using the oracle for Problem~\ref{prob:4}.\\
	 \textbf{Step 3:} Verify that there exists an SPE $\tau$ such that $p \leq (\Cost_i(\outcome{\tau}{x_0}))_{i\in\Pi} \leq p$ (Problem~\ref{prob:decisionThreshold})\\

	We obtain an SPE $\tau$ such that $(\Cost_i(\outcome{\tau}{x_0}))_{i\in\Pi} = p$ is Pareto optimal. By Proposition~\ref{prop:particularPareto}, we have that for all $i \in \Visit(\outcome{\tau}{x_0})$, $\Cost_i(\outcome{\tau}{x_0}) \leq |V|\cdot|\Pi|$.	Thus, by Proposition~\ref{prop:SPEmemory}, there exists an  SPE $\sigma$ with the same cost profile than $\tau$ but with memory in $\mathcal{O}(2^{|\Pi|} \cdot |\Pi|\cdot |V|^{(|\Pi|+1)\cdot(|\Pi|+|V|)+1})$.\\
	
	The PSPACE-hardness is due to a polynomial reduction from the QBF problem which is PSPACE-complete. The philosophy of the proof is the same than the one for the qualitative setting (see proof of Proposition~\ref{prop:qualSPEPareto}) but we have to put a sufficiently high weight on the edge between the initial vertex and the vertex $\perp$.\qed
\end{proof}

\section{Qualitative reachability games}

\subsection{Additional preliminaries}

\subsubsection{Qualitative reachability games}

All along this section we focus on \emph{qualitative reachability games}. Unlike quantitative reachability games, the arena is equipped with a gain function profile $(\Gain_i)_{i\in\Pi}$ such that for all $i\in \Pi$, $\Gain_i: \Plays \rightarrow \mathbb{N}\cup \{+\infty\}$ is a \emph{gain function} which assigns a gain to each play $\rho$ for Player $i$. We also say that the play $\rho$ has gain profile $(\Gain_i(\rho))_{i\in \Pi}$ and similarly if we consider the outcome of the strategy profile $\sigma$  from $v_0$, we say that $\sigma$ has gain profile $(\Gain_i(\outcome{\sigma}{v_0}))_{i\in\Pi}$.

\begin{definition} \label{def:qualitativeGame}
	A \emph{qualitative reachability game} $\mathcal{G} = (\mathcal{A}, (\Gain_i)_{i\in \Pi}, (F_i)_{i\in \Pi})$ is a game enhanced with a target set $F_i \subseteq V$. For all $i \in \Pi$, the gain function $\CostOrGain_i = \Gain_i$ is defined as follows: for all $\rho= \rho_0\rho_1\ldots \in \Plays$: $\Gain_i(\rho) = 1$ if there exists $k\in \mathbb{N}$ such that $\rho_k \in F_i$ and $\Gain_i(\rho) = 0$ otherwise.
\end{definition}

In this particular setting, players only aim to reach their target set but do not take into account the number of steps it takes. Player $i$ receives a gain of $1$ if $\rho$ visits his target set $F_i$, and a gain of $0$ otherwise. Thus each player~$i$ wants to \textbf{maximize} his gain.

\subsubsection{Solution concepts}

For qualitative reachability games, it is easy to recover the definitions of NE and SPE defined in Section~\ref{section:solutionConcepts}  by reversing the inequality and replacing cost functions by gain functions, as players want to maximize their gain instead of minimizing their cost.
This leads to the following Lemma. 

\begin{lemma} \label{lem:QuantitativeToQualitative}
Let $(\mathcal{G},v_0)$ be an initialized quantitative reachability game and $\sigma$ be a strategy profile. Consider the related qualitative reachability game $\mathcal{G}'$ with the same arena $\arena$ and target sets $(F_i)_{i \in \Pi}$, but the gain functions $(\Gain_i)_{i \in \Pi}$. Then if $\sigma$ is an NE (resp. SPE) in $(\mathcal{G},v_0)$, then $\sigma$ is also an NE (resp. SPE) in $(\mathcal{G}',v_0)$.
\end{lemma}
 
Thus, as it is proved that there always exists an SPE (and thus an NE) in a quantitative reachability game, there always exists one in a qualitative reachability game. 

\begin{theorem} \label{thm:SPEexistqual}
In every initialized qualitative reachability game, there always exists an SPE, and thus also an NE.
\end{theorem}

\subsubsection{Studied problems}

 In case of qualitative reachability, as for quantitative reachability game we are interested in a solution that fulfills certain requirements. For example, we would like to know whether there exists a solution such that a maximum number of players visit their target sets.

Let $(\mathcal{G},v_0)$ be an initialized  qualitative reachability game with $\mathcal{G} =$ \\$ (\mathcal{A}, (\Gain_i)_{i\in \Pi}, (F_i)_{i\in \Pi})$. Given $\rho \in \Plays(v_0)$, we denote by $\Visit(\rho)$ the set of players $i$ such that $\rho$ visits $F_i$, that is, $\Visit(\rho) = \{ i \in \Pi \mid Gain_i(\rho) =  1 \}$. The \emph{social welfare} $\SW(\rho)$ of $\rho$ is the size of $\Visit(\rho)$. Let $P \subseteq \{0,1\}^{|\Pi|}$ be the set of all gain profiles $p = (\Gain_i(\rho))_{i\in\Pi}$, with $\rho \in \Plays(v_0)$. A cost profile $p \in P$ is called \emph{Pareto-optimal in $\Plays(v_0)$} if it is maximal in $P$ with respect to the componentwise ordering $\leq$ on $P$. Notice that if there exists $\rho$ with $\Visit(\rho) = \Pi$, then its social welfare is the largest possible and there exists a unique Pareto optimal gain profile equal to $(1,1,\ldots,1)$. Notice also that certain target sets $F_i$ might be empty or not reachable from the initial vertex $v_0$. Hence in this case, the best that we can hope is a (unique) Pareto optimal gain profile $p$ such that $p_i = 1$ if and only if $F_i$ is reachable\footnote{Notice that if $F_i$ is reachable from $v_0$, then it is necessarily not empty.} from $v_0$.

\begin{problem}[Threshold decision problem]
\label{prob:decisionThresholdQual}
Given an initialized qualitative reachability game $(\mathcal{G},v_0)$, given two thresholds $x,y \in \{0,1\}^{|\Pi|}$, decide whether there exists a solution $\sigma$ such that $x \leq (\Gain_i(\rho))_{i \in \Pi} \leq y$.
\end{problem}

Imposing a lower bound $x_i = 1$ means that player~$i$ has to visit his target set whereas imposing an upper bound $y_i = 0$ means that player~$i$ cannot visit his target set.

Unlike quantitative reachability, social welfare in qualitative reachability games only aims to maximize the number of players who visit their target set.

\begin{problem}[Social welfare decision problem]
\label{prob:decisionWelfareQual}
Given an initialized qualitative reachability game $(\mathcal{G},v_0)$, given a threshold $k \in \{0, \ldots, |\Pi| \}$, decide whether there exists a solution $\sigma$ such that $\SW(\outcome{\sigma}{v_0}) \geq k$.
\end{problem}

Let us now state the last studied problem for qualitative reachability games.

\begin{problem}[Pareto-optimal decision problem]
\label{prob:decisionParetoQual}
Given an initialized qualitative reachability game $(\mathcal{G},v_0)$ decide whether there exists a solution $\sigma$ in $(\mathcal{G},v_0)$ such that $(\Gain_i(\outcome{\sigma}{v_0}))_{i\in\Pi}$ is Pareto optimal in $\Plays(v_0)$.
\end{problem}

This problem has some connections with the two previous ones. For instance in case of qualitative reachability, suppose there exists a play in $\Plays(v_0)$ that visits all target sets. As already explained, there is only one Pareto-optimal gain $(1,\ldots,1)$. Asking for the existence of a solution $\sigma$ such that $(\Gain_i(\outcome{\sigma}{v_0}))_{i\in\Pi}$ is Pareto-optimal is equivalent to asking for the existence of a solution $\sigma$ such that $\Gain_i(\outcome{\sigma}{v_0}))_{i\in\Pi}  \geq (1,\ldots,1)$ (see Problem~\ref{prob:decisionThresholdQual}), or such that $\SW(\outcome{\sigma}{v_0}) \geq |\Pi|$ (see Problem~\ref{prob:decisionWelfareQual}).


\subsection{Existence problem}

The following Theorem is a direct consequence of Theorem~\ref{thm:visitAll} and Lemma~\ref{lem:QuantitativeToQualitative}.

\begin{theorem} \label{thm:visitAllQual}
Let $(\mathcal{G},v_0)$ be an initialized qualitative reachability game such that its arena $\arena$ is strongly connected. Then there exists an SPE $\sigma$ (and thus an NE) such that its outcome $\outcome{\sigma}{v_0}$ visits all target sets $F_i$, $i \in \Pi$, that are non empty.
\end{theorem}

Let us comment this result. For this family of games, the answer to Problem~\ref{prob:decisionThresholdQual} is always positive for particular thresholds. Take thresholds $x,y$ such that $x_i = 1$ (and thus $y_i = 1$) if and only if $F_i \neq \emptyset$. The answer to Problem~\ref{prob:decisionWelfareQual} is also always positive for threshold $k = |\{ i \mid F_i \neq \emptyset \}|$. Finally, the answer to Problem~\ref{prob:decisionPareto} is also always positive since there exists an unique Pareto optimal gain profile $p$ such that $p_i = 1$ if and only if $F_i \neq \emptyset$.\\

Recall that we explained before why it was enough to prove Theorem~\ref{thm:visitAll} for SPEs and for quantitative reachability games only. Notice that in case of qualitative reachability games, there exists a simpler construction of the required NE or SPE. Indeed, as the arena is strongly connected, there exists a play $\rho \in \Plays(v_0)$ that visits all non empty target sets. \emph{(i)} Hence to get an NE, construct a strategy profile $\sigma$ in $(\mathcal{G},v_0)$ such that $\outcome{\sigma}{v_0} = \rho$. As the gain profile of $\sigma$ is the best that each player can hope, no player has an incentive to deviate and $\sigma$ is then an NE. \emph{(ii)} The construction is a little more complex to get an SPE. We again construct a strategy profile $\sigma$ in $(\mathcal{G},v_0)$ such that $\outcome{\sigma}{v_0} = \rho$, and inductively extend its construction to all subgames $(\rest{\mathcal{G}}{h},v)$ as follows. Assume that $\rest{\sigma}{h}$ is not yet constructed, then extend the construction of $\sigma$ such that $\rest{\sigma}{h} = g\rho$ for some $gv_0$ starting in $v$ and ending in $v_0$ (such a history $gv_0$ exists because the arena is strongly connected). In this way, the outcome of $\rest{\sigma}{h}$ in each subgame $(\mathcal{G},v)$ has gain profile $(1,\ldots, 1)$ and no player has an incentive to deviate. It follows that $\sigma$ is an SPE.

The next theorem states that Problem~\ref{prob:decisionPareto} has a positive answer for all \emph{qualitative} reachability games with a number of players \emph{limited to two}, and that this existence result cannot be extended to three players.

\begin{theorem} \label{thm:existenceParetoQual}
Let $(\mathcal{G},v_0)$ be an initialized qualitative reachability game,
\begin{itemize}
\item  Let $(\mathcal{G},v_0)$ be an initialized qualitative reachability game such that $|\Pi| = 2$, there exists an SPE $\sigma$ (and thus an NE) with a gain profile that is Pareto optimal in $\Plays(v_0)$.
\item There exists an initialized qualitative reachability games with $|\Pi| = 3$  that has no NE with a gain profile  that is Pareto optimal in $\Plays(v_0)$.
\end{itemize}
\end{theorem}

Let us focus on the proof of Theorem~\ref{thm:existenceParetoQual} which is based on the next lemma, which is interesting in its own right.

\begin{lemma} \label{lem:pareto}
Let $(\mathcal{G},v_0)$ be an initialized qualitative reachability game. Let $p$ be a gain profile equal to $(0,0,\ldots,0)$ or $(1,1,\ldots,1)$. If $p$ is Pareto-optimal\footnote{$(1,1,\ldots,1)$ is trivially Pareto-optimal.} in $\Plays(v_0)$, then there exists an SPE $\sigma$  with gain profile $p$.
\end{lemma}

\begin{proof}
The case $p = (0,0,\ldots,0)$ is easy to solve. By Pareto-optimality, all plays in $\Plays(v_0)$ have gain profile $p$. Hence every strategy profile $\sigma$ is trivially an SPE with gain profile $p$.  Let us turn to case $p = (1,1,\ldots,1)$ and let $\rho = \rho_0\rho_1 \ldots \in \Plays(v_0)$ with gain profile $p$. By Theorem 2.1 in \cite{DBLP:journals/corr/abs-1205-6346}\footnote{Notice that we cannot apply Theorem~\ref{thm:visitAll} since the arena is not necessarily strongly connected.}, there exists an SPE $\sigma$ in $(\mathcal{G},v_0)$. If $(\Gain_i(\outcome{\sigma}{v_0}))_{i \in \Pi} = p$, we are done. Otherwise let us show how to modify $\sigma$ into another SPE $\tau$ with outcome $\rho$ and thus with gain profile $p$. Let $h \in \Hist_i(v_0)$, $i \in \Pi$,
\begin{itemize}
\item if $h$ is a prefix of $\rho$, then $\tau_i(h) = \rho_{|h|+1}$,
\item otherwise, $\tau_i(h) = \sigma_i(h)$.
\end{itemize}
Let us prove that $\tau$ is an SPE. Clearly for each history $hv$ that is not a prefix of $\rho$, $\rest{\tau}{h} = \rest{\sigma}{h}$ is an NE in the subgame $(\rest{\mathcal{G}}{h},v)$. So let $hv = \rho_0 \ldots \rho_k$. As $\outcome{\rest{\tau}{h}}{v}$ has gain profile $(1,1,\ldots, 1)$ in $(\rest{\mathcal{G}}{h},v)$, player~$i$ such that $v \in V_i$ has no incentive to deviate, and then $\rest{\tau}{h}$ is also an NE in $(\rest{\mathcal{G}}{h},v)$.
\qed\end{proof}

\begin{proof}[of Theorem~\ref{thm:existenceParetoQual}]
We begin with the first item. There are three cases to study: either the unique Pareto-optimal gain profile of $\Plays(v_0)$ is equal to $(0,0)$, or it is equal to $(1,1)$, or there are one or two Pareto-optimal gain profiles that belong to $\{(0,1),(1,0)\}$. In the first two cases, we get the required SPE by Lemma~\ref{lem:pareto}. Hence it remains to treat the last case. From Lemma~\ref{lem:visitOne}, we know that there exists an SPE in $(\mathcal{G},v_0)$ whose outcome $\rho$ visits a least one target set $F_i$, $i \in \{1,2\}$. Therefore the gain profile of $\rho$ is either equal to $(0,1)$ or $(1,0)$ as required.

\begin{figure}
 \centering
 \scalebox{0.8}{\begin{tikzpicture}[-latex, auto, node distance = 1 cm and 1 cm, on grid, semithick, round/.style = {draw,circle}, sq/.style = {draw,rectangle}, di/.style = {draw,diamond}]
  \node[round] (0) at (0,0) {$v_0$};
  \node[sq] (1) at (-1.5,0) {$v_1$};
  \node[di] (2) at (1.5,0) {$v_2$};
  \node[round] (3) at (-2.5,-0.5) {$v_3$};
  \node[round] (4) at (-2.5,0.5) {$v_4$};
  \node[round] (5) at (2.5,-0.5) {$v_5$};
  \node[round] (6) at (2.5,0.5) {$v_6$};
  \node[text = black] (3payoff) [left = of 3,xshift=-0.5cm] {$(0,1,0)$};
  \node[text = black] (4payoff) [left = of 4,xshift=-0.5cm] {$(1,0,1)$};
  \node[text = black] (5payoff) [right = of 5,xshift=0.5cm] {$(1,1,0)$};
  \node[text = black] (6payoff) [right = of 6,xshift=0.5cm] {$(0,0,1)$};
  \draw[->] (0) -- (1);
  \draw[->] (0) -- (2);
  \draw[->] (1) -- (3);
  \draw[->] (1) -- (4);
  \draw[->] (2) -- (5);
  \draw[->] (2) -- (6);
  \draw[->] (3) to [loop left] (3);
  \draw[->] (4) to [loop left] (4);
  \draw[->] (5) to [loop right] (5);
  \draw[->] (6) to [loop right] (6);
 \end{tikzpicture}}
\caption{A qualitative reachability game that has no NE with a gain profile that is Pareto-optimal}
\label{fig:3players}
\end{figure}
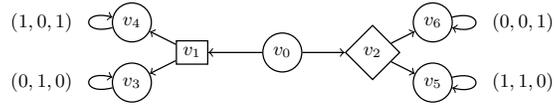

For the second item, consider the initialized qualitative reachability game $(\mathcal{G},v_0)$ of Figure~\ref{fig:3players}. We have three players such that player~$3$ owns diamond vertices. Moreover, $F_1 = \{v_4,v_5\}$, $F_2 = \{v_3,v_5\}$, and $F_3 = \{v_4,v_6 \}$. There are four plays in $\Plays(v_0)$ whose gain profile is indicated below each of them. The set of Pareto-optimal gain profiles in $\Plays(v_0)$ is equal to $\{(1,0,1),(1,1,0)\}$. Consider a strategy profile $\sigma$ with outcome $v_0v_1v_4^\omega$ and gain profile $(1,0,1)$. Then it is not an NE because player~$2$ has a profitable deviation by going from $v_1$ to $v_3$ (instead of $v_4$).  Similarly the strategy profile $\sigma$ with outcome $v_0v_2v_5^\omega$ and gain profile $(1,1,0)$ is not an NE. Therefore there is no NE in $(\mathcal{G},v_0)$ with a gain profile that is Pareto-optimal. 

 \qed
\end{proof}

\subsection{Decision problems}

\subsubsection{$\lambda$-consistent play}

The notion of $\lambda$-consistent play explained in Section~\ref{section:consistentPlay} can easily be adapted for qualitative reachability games. In the qualitative setting, $\lambda(v)$ allows us to know if the player who owns $v$ can ensure to reach his target set from $v$ or not.

\subsubsection{Lassoes with polynomial size}

\begin{lemma}
	\label{lem:LassoesQual}
	Let $(\mathcal{G},v_0)$ be a qualitative reachability game and $\rho \in \Plays$ be a play.
	\begin{itemize}
		\item If $\rho'$ is the play obtained by applying (P1) or (P2) on $\rho$, then for all $i\in \Pi$, $\Gain_i(\rho') = \Gain_i(\rho)$.
		\item Applying (P1) until it is no longer possible and then (P2), leads to a lasso $\rho'$ with size at most $(|\Pi|+1)\cdot |V|$.			
	\end{itemize}
\end{lemma}

And the property on $\lambda$-consistence after applying (P1) or (P2) also holds.

\begin{lemma}
	\label{lem:LassoesConsisQual}
	Let $(\mathcal{G},v_0)$ be a  qualitative reachability game and $\rho \in \Plays$ be a $\lambda$-consistent play for a given labeling function $\lambda$. If $\rho'$ is the play obtained by applying (P1) or (P2) on $\rho$, then $\rho'$ is $\lambda$-consistent.
\end{lemma}

\subsubsection{Results on Nash equilibria}

\paragraph{Characterization of outcomes of NEs} The characterization provides in Section~\ref{subsection:charac} for NEs remains true for qualitative reachability games with some modifications: using the notion of qualitative coalitional game in place of quantitative coalitional game, replacing $\Cost$ by $\Gain$ in the characterization, reversing the inequalities, $\ldots$

Notions of coalitional game, value and optimal strategies can be rewritten if $\mathcal{G}$ is a qualitative reachability game and we obtain $|\Pi|$ (qualitative) coalitional games. Qualitative coalitional games are also determined, there exist positional optimal strategies for both players and the values can be computed in polynomial time~\cite{10.1007/3-540-59042-0_57}.

\paragraph{Complexity and memory results}

\begin{theorem}
	\label{thm:NEComplexityQual}
	Let $(\mathcal{G},v_0)$ be a qualitative reachability game, for NE: Problem~\ref{prob:decisionThresholdQual} and Problem~\ref{prob:decisionWelfareQual} are NP-complete while Problem~\ref{prob:decisionParetoQual} is NP-hard and belongs to $\Sigma^P_2$.	
\end{theorem}

\begin{theorem}
	\label{thm:NEMemoryQual}
		Let $(\mathcal{G},v_0)$ be a qualitative reachability game, for NE: for each decision problem, if its answer is positive, then there exists a strategy profile $\sigma$ with memory in  $\mathcal{O}((|\Pi|+1)\cdot |V|)$ which satisfies the conditions.
\end{theorem}

These two theorems are due to Propositions~\ref{prop:qualNEThreshold},\ref{prop:qualNEWelfare} and~\ref{prop:qualNEPareto}

\begin{proposition}[\cite{condurache_et_al:LIPIcs:2016:6256}]
	\label{prop:qualNEThreshold}
	Let $(\mathcal{G},v_0)$ be a qualitative reachability game, for NE: 
	\begin{itemize}
	\item Problem~\ref{prob:decisionThresholdQual} is NP-complete.
	\item If the answer to this decision problem is positive, then there exists a strategy profile with memory in $\mathcal{O}((|\Pi|+1)\cdot |V|)$ which satisfies the constraints.
\end{itemize}
\end{proposition}
	
\begin{proposition}
	\label{prop:qualNEWelfare}
	Let $(\mathcal{G},v_0)$ be a qualitative reachability game. 
	\begin{itemize}
		\item For NE, Problem~\ref{prob:decisionWelfareQual} is NP-complete.
		\item If the answer to this decision problem is positive, then there exists a strategy profile with memory in  $\mathcal{O}((|\Pi|+1)\cdot |V|)$ which satisfies the constraints. 
	\end{itemize}
\end{proposition}

\begin{proof}
	
	Let $(\mathcal{G},v_0)$ be a qualitative reachability game and let $k \in \{0, \ldots, |\Pi| \}$ be a threshold.
	
	 The NP-algorithm works as follows:\\	
		\textbf{Step 1:} Guess $p \in \{0,1\}^{|\Pi|}$ a gain profile (sufficient thanks to Lemma~\ref{lem:LassoesQual});\\
		\textbf{Step 2:} Let $\mathcal{R} = \{ i \in \Pi \mid p_i = 1 \}$ be the set of players who have a gain equal to 1 in the gain profile $p$, verify if $|\mathcal{R}| \geq k$;\\
		\textbf{Step 3:} Verify that there exists an NE $\sigma$ in $(\mathcal{G},v_0)$ such that $(\Gain_i(\outcome{\sigma}{v_0}))_{i\in\Pi} \geq p.$\\
	
	As if $(\Gain_i(\outcome{\sigma}{v_0}))_{i\in\Pi} \geq p$ then $|\Visit(\outcome{\sigma}{v_0})| \geq |\mathcal{R}| \geq k$, and thanks to Proposition~\ref{prop:qualNEThreshold}, this algorithm is an NP-algorithm to solve Problem~\ref{prob:decisionWelfareQual} for Nash equilibrium. Moreover, if there exists an NE which fulfills the constraints there exists one with a polynomial memory which fulfills the constraints too.\\
	
The NP-hardness is due to a polynomial reduction from the SAT problem in the same philosophy than the one for Problem~\ref{prob:decisionThreshold} (see Proposition~\ref{prop:compConstNashHard}) for NE in quantitative reachability games: the SAT-formula is satisfiable if and only if there exists an NE with a social welfare $\geq |\Pi|$.  \qed
	
\end{proof}

\begin{proposition}
	\label{prop:qualNEPareto}
		Let $(\mathcal{G},v_0)$ be a qualitative reachability game, for NE:
		
		\begin{itemize}
			\item Problem~\ref{prob:decisionParetoQual} belongs to $\Sigma^P_2$ and is NP-hard.
			\item If the answer of this decision problem is positive, then there exists a strategy profile with memory in  $\mathcal{O}((|\Pi|+1)\cdot |V|)$ which satisfies the constraints. 
		\end{itemize}
\end{proposition}

The proof of Proposition~\ref{prop:qualNEPareto} relies on the complexity of the variant of Problem~\ref{prob:4} for qualitative reachability games. 

\begin{problem}
	\label{prob:4qual}
	Given a qualitative reachability game $(\mathcal{G},v_0)$ and a lasso $\rho \in \Plays$, we want to decide if $(\Gain_i(\rho))_{i\in\Pi}$ is  Pareto optimal  in $\Plays(v_0)$.
\end{problem} 

This decision problem belongs to co-NP.

\begin{lemma}
	\label{lemma:prob4qual}
	Problem~\ref{prob:4qual} belongs to co-NP for qualitative reachability games.
\end{lemma}
\begin{proof}
	
	The same proof as for quantitative reachability games holds by using $\Gain$ in place of $\Cost$ and reversing the inequalities. \qed
\end{proof}
\begin{proof}[of Proposition~\ref{prop:qualNEPareto}]
	We can provide the following $\Sigma^P_2$-algorithm for Problem~\ref{prob:decisionPareto}, given an oracle for Problem~\ref{prob:4qual}:\\
	 \textbf{Step 1:} Guess a play $\rho$ as a lasso of size $(|\Pi|+1)\cdot|V|$ (sufficient thanks to Lemma~\ref{lem:LassoesQual});\\
	 \textbf{Step 2:} Check that $(\Gain_i(\rho))_{i\in \Pi}$ is Pareto optimal in $\Plays(v_0)$ using the oracle for Problem~\ref{prob:4qual}.\\
	 \textbf{Step 3:} Check that $\rho$ is the outcome of a Nash equilibrium using the characterization.\\
	
	The NP-hardness is due to a polynomial reduction from the SAT problem that is NP-complete. Let $X = \{x_1, \ldots, x_m \}$ be the set of variables and $\psi = C_1\wedge\ldots \wedge C_n$ is a Boolean formula in CNF over $X$ and equals to the conjunction of the clauses $C_1, \ldots, C_n$. This problem is to decide if the formula $\psi$ is true. Such a formula is true if there exists a valuation $\mathcal{I} : X \rightarrow \{0,1\}$ such that the valuation of $X$ with respect to $\mathcal{I}$ evaluates $\psi$ to true.

	We build the following qualitative reachability game $\mathcal{G}_{\psi}= (\mathcal{A}, (\Gain_i)_{i \in \Pi}, (F_i)_{i\in \Pi})$ where $\mathcal{A} = (\Pi, V, E, (V_i)_{i \in \Pi})$:
	\begin{itemize}
		\item the arena $\arena$ is depicted in Figure~\ref{fig:redFormToGamePareto} where the set of players $\Pi = \{1, \ldots, n, n+1 \}$ has $n+1$ players: one by clause (players $1$ to $n$) and an additional one. $V_{n+1}$ is depicted by squared vertices and for all $1 \leq i \leq n$, $V_i = \{P_i \}$;
		\item $F_{n+1} = \{ T_w, \perp \}$;
		\item for all $1 \leq i \leq n$, $F_i = \{0_x \mid \neg x \in C_i \} \cup \{ 1_x \mid x \in C_i \}\cup \{ T_\ell \}.$
	\end{itemize}

	The game $\mathcal{G}_{\psi}$ can be build from $\psi$ in polynomial time. Let us show that $\psi$ is true if and only if there exists an NE $ \sigma $ in $(\mathcal{G}_\psi,x_1)$ such that $\Gain(\outcome{\sigma}{x_1})$ is Pareto optimal in $\Plays(x_1)$.\\

	$(\Rightarrow)$ Suppose that $\psi$ is true. Then there exists $\mathcal{I}: X \rightarrow \{0,1\}$ such that the valuation of $\psi$ with respect to $\mathcal{I}$ evaluates $\psi$ to true.
	
	Let us consider $\sigma$ defined in this way:
	
	\begin{itemize}
		\item for all $hv \in \Hist_{n+1}(x_1)$, $\sigma_{n+1}(hv) = \begin{cases} \mathcal{I}(v) & \text{ if } v = x \text{ with } x \in X \\
		 																		x_{k+1} & \text{ if } v = 0_{x_k} \text{ or } 1_{x_k} \text{with } 0 \leq k \leq m-1 \\
		P_1 & \text{if } v = 0_{x_m} \text{ or } 1_{x_m} \\
		T_\ell & \text{if } v = T_\ell \\
		T_w & \text{if } v = T_w\\
		\perp & \text{if } v = \perp\end{cases}$.
		\item for all $1 \leq i \leq n-1 $ (resp. for $i = n$), for all $hv \in \Hist_i(x_1)$ $\sigma_i(hv) = P_{i+1} $ (resp. $T_w$) if $hv$ is consistent with $\sigma_{n+1}$ and $\sigma_i(hv)= T_\ell$ otherwise. 	
		
	\end{itemize}
	
	We now prove that $\sigma$ is an NE. It is clear that $\outcome{\sigma}{x_1}$ is of the forme $hP_1P_2\ldots P_n T_w^\omega$ and as $\sigma_{n+1}$ corresponds to the valuation $\mathcal{I}$ which evaluates $\psi$ to true, we have that  for all $ i\in \Pi$, $\Gain_{i}(\outcome{\sigma}{x_1})= 1$. Obviously no player has a profitable deviation and $(\Gain_i(\outcome{\sigma}{x_1}))_{i\in\Pi}$.\\
	
	$(\Leftarrow)$ We will prove its contrapositive. Assume that $\psi$ is not true. As $\psi$ is not true, there exists a player $i \in \Pi$ such that if he wants to win (obtain a gain of 1) he has to go to $T_{\ell}$. However, in this way Player $n+1$ does not reach his target set and his only possibility to do so is to go to $\perp$.  It means that the only possible outcome for an NE is $x_1\perp^\omega$ with gain profile $(0, \ldots,0, 1)$ which is not Pareto optimal in $\Plays(x_1)$. Indeed, as at least one clause can be evaluated to true, assume it is clause $C_i$ (with $1\leq i \leq n$), there exists a play in $\Plays(x_1)$ such that Player $i$ visits his target set and which ends in $T_w$.

	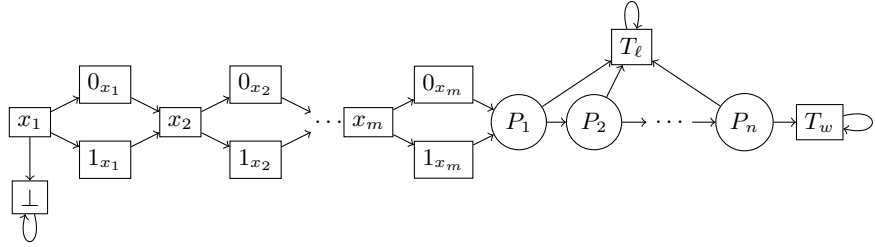
\begin{figure}

		\centering
		\begin{tikzpicture}
			\node[draw] (x1) at (0,0){$x_1$};
			\node[draw] (x10) at (1, 0.5){$0_{x_1}$};
			\node[draw] (x11) at (1, -0.5){$1_{x_1}$};
			
			\node[draw] (p) at (0,-1){$\perp$};
			
			\draw[->] (x1) to (x10);
			\draw[->] (x1) to (x11);
			\draw[->] (x1) to (p);
			\draw[->] (p) to [loop below] (p);
						
			\node[draw] (x2) at (2,0){$x_2$};
			\node[draw] (x20) at (3,0.5){$0_{x_2}$};
			\node[draw] (x21) at (3,-0.5){$1_{x_2}$};
			
			\draw[->] (x10) to (x2);
			\draw[->] (x11) to (x2);
			\draw[->] (x2) to (x20);
			\draw[->] (x2) to (x21);
			
			\node (int1) at (4,0){\ldots};
			\draw[->] (x20) to (int1);
			\draw[->] (x21) to (int1);
			
			\node[draw] (xm) at (4.5,0){$x_m$};
			\node[draw] (xm0) at (5.5,0.5){$0_{x_m}$};
			\node[draw] (xm1) at (5.5,-0.5){$1_{x_m}$};
			
			\draw[->] (xm) to (xm0);
			\draw[->] (xm) to (xm1);
			
			\node[draw,circle] (P1) at (6.5,0){$P_1$};
			
			\draw[->] (xm1) to (P1);
			\draw[->] (xm0) to (P1);
			
			\node[draw,circle] (P2) at (7.5,0){$P_2$};
			
			\draw[->] (P1) to (P2);
			
			\node (int2) at (8.5,0){$\ldots$};
			\draw[->] (P2) to (int2);
			
			\node[draw,circle] (Pn) at (9.5,0){$P_n$};
			\draw[->] (int2) to (Pn);
			
			\node[draw] (Tw) at (10.5,0){$T_{w}$};
			\draw[->] (Pn) to (Tw);
			
			\draw[->] (Tw) to [loop right] (Tw);
			
			\node[draw] (Tl) at (8,1){$T_{\ell}$};

			\draw[->] (Tl) to [loop above] (Tl);
			
			\draw[->] (P1) to (Tl);
			\draw[->] (P2) to (Tl);
			\draw[->] (Pn) to (Tl);			
		\end{tikzpicture}
			\caption{Reduction from the formula $\psi$ to the quantitative reachability game $\mathcal{G}_{\psi}$}
			\label{fig:redFormToGamePareto}
	\end{figure}

	\qed
\end{proof}

\subsubsection{Results on subgame perfect equilibria}

\paragraph{Complexity and memory results}

\begin{theorem}
	\label{thm:SPEComplexityQual}
	Let $(\mathcal{G},v_0)$ be a qualitative reachability game, for SPE:
	Problem~\ref{prob:decisionThresholdQual}, Problem~\ref{prob:decisionWelfareQual} and Problem~\ref{prob:decisionParetoQual} are PSPACE-complete.
\end{theorem}

\begin{theorem}
	\label{thm:SPEMemoryQual}	
	Let $(\mathcal{G},v_0)$ be a qualitative reachability game, for SPE: for each decision problem, if the answer is positive, there exists a strategy profile $ \sigma$ with memory in $\mathcal{O}(|V|^3\cdot |\Pi| \cdot 2^{3 \cdot |\Pi|})$ which satisfies the conditions.
\end{theorem}

These two theorems are due to Propositions~\ref{prop:qualSPEThreshold},\ref{prop:qualSPEWelfare} and~\ref{prop:qualSPEPareto}.

\begin{proposition}[\cite{DBLP:journals/corr/abs-1809-03888}]
	\label{prop:qualSPEThreshold}
	Let $(\mathcal{G},v_0)$ be a qualitative reachability game.
	\begin{itemize}
		\item For SPE, Problem~\ref{prob:decisionThresholdQual} is PSPACE-complete.
	\item If the answer of the decision problem is positive, there exists a strategy profile with memory in  $\mathcal{O}(|V|^3\cdot |\Pi| \cdot 2^{3 \cdot |\Pi|})$ which satisfies the constraints.
	\end{itemize}
\end{proposition}

\begin{proposition}
	\label{prop:qualSPEWelfare}
	Let $(\mathcal{G},v_0)$ be a qualitative reachability game.
	\begin{itemize}
		\item For SPE, Problem~\ref{prob:decisionWelfareQual} is PSPACE-complete.
	\item If the answer of the decision problem is positive, there exists a strategy profile with memory in  $\mathcal{O}(|V|^3\cdot |\Pi| \cdot 2^{3 \cdot |\Pi|})$ which satisfies the constraints.
	\end{itemize}
\end{proposition}

\begin{proof}
	
	Let $(\mathcal{G},v_0)$ be a qualitative reachability game and let $k \in \{0, \ldots, |\Pi| \}$ be a threshold.
	
	 The PSPACE-algorithm works as follows:\\	
		\textbf{Step 1:} Guess $p \in \{0,1\}^{|\Pi|}$ a gain profile;\\
		\textbf{Step 2:} Let $\mathcal{R} = \{ i \in \Pi \mid p_i = 1 \}$ be the set of players who have a cost equal to 1 in the cost profile $p$, verify if $|\mathcal{R}| \geq k$;\\
		\textbf{Step 3:} Verify that there exists an SPE $\sigma$ in $(\mathcal{G},v_0)$ such that $(\Gain_i(\outcome{\sigma}{v_0}))_{i\in\Pi} \geq p.$\\
	
	As if $(\Gain_i(\outcome{\sigma}{v_0}))_{i\in\Pi} \geq p$ then $|\Visit(\outcome{\sigma}{v_0})| \geq |\mathcal{R}| \geq k$, and thanks to Proposition~\ref{prop:qualSPEThreshold}, this algorithm is a PSPACE-algorithm to solve Problem~\ref{prob:decisionWelfareQual} for SPE. Moreover, if there exists an SPE which fulfills the constraints there exists one with a finite memory which fulfills the constraints too.\\
	
The PSPACE-hardness is due to a polynomial reduction from the QBF problem (which is PSPACE-complete) in the same philosophy than the one for Problem~\ref{prob:decisionWelfare} for SPE in quantitative reachability games: the fully quantified formula is satisfiable if and only if the exists an SPE with a social welfare $\geq |\Pi|-1$. Notice that in the qualitative setting, there is no weight on the edge $(q_1,\perp)$ in the arena depicted in Figure~\ref{figure:reachPSPACEh}.  \qed
\end{proof}

\begin{proposition}
		\label{prop:qualSPEPareto}
	Let $(\mathcal{G},v_0)$ be a qualitative reachability game,
		\begin{itemize}
			\item For SPE, Problem~\ref{prob:decisionParetoQual} is PSPACE-complete.
			\item If the answer of the decision problem is positive, there exists a strategy profile with memory in  $\mathcal{O}(|V|^3\cdot |\Pi| \cdot 2^{3 \cdot |\Pi|})$ which satisfies the constraints.
		\end{itemize}
\end{proposition}

\begin{proof}
	We can provide the following PSPACE-algorithm for Problem~\ref{prob:decisionParetoQual}, given an oracle for Problem~\ref{prob:4qual}:\\
	 \textbf{Step 1:} Guess a payoff $p \in \{0,1\}^{|\Pi|}$.\\
	 \textbf{Step 2:} Check that $p$ is Pareto optimal in $\Plays(v_0)$ using the oracle for Problem~\ref{prob:4qual}.\\
	 \textbf{Step 3:} Check if there exists an SPE $\sigma$, such that $p \leq (\Gain_i(\outcome{\sigma}{v_0}))_{i\in\Pi} \leq p$.\\
	
	By Lemma~\ref{lemma:prob4qual} and Proposition~\ref{prop:qualSPEThreshold}, Step 2 is done in co-NP  and Step 3 is done in PSPACE. It leads to a PSPACE-algorithm.

The PSPACE-hardness is due to a polynomial reduction from the QBF problem (which is PSPACE-complete) in the same philosophy than the one for Problem~\ref{prob:decisionWelfare} for SPE in quantitative reachability games: the fully quantified formula is satisfiable if and only if the exists an SPE with a gain profile which is Pareto optimal. Notice that in the qualitative setting, there is no weight on the edge $(q_1,\perp)$ in the arena depicted in Figure~\ref{figure:reachPSPACEh}.

 For the implication ``$\Leftarrow$'', we assume that we have an SPE $\sigma$ such that $(\Gain_i(\outcome{\sigma}{q_1}))_{i\in \Pi}$ is Pareto optimal in $\Plays(q_1)$  and we want to prove that the fully quantified formula $\psi$ is satisfied. In the game $\mathcal{G}_{\psi}$, there is only two Pareto optimal gain profile: $(g_1,\ldots,g_n,1,0)$ and $(g_1,\ldots,g_n,0,1)$ where $g_i = 1$ if the clause $C_i$ is satisfiable in $\psi$ and $g_i=0$ otherwise. One can prove that the gain profile $(g_1,\ldots,g_n,0,1)$ cannot be the gain profile of an SPE as in this gain profile Player $n+1$ has a gain of $0$ and Player $n+1$ can chose to go in $\perp$ to ensure a gain of $1$. Thus, we know that $\sigma$ is an SPE with gain profile $(g_1,\ldots,g_n,1,0)$ and in particular such that Player $n+1$ visits his target set. Then, the same kind of arguments used in the proof of Problem~\ref{prob:decisionThreshold} (see~\cite{DBLP:journals/corr/abs-1905.00784}) can be used. \qed
\end{proof}

\end{document}